\documentclass[letterpaper,11pt]{article}

\usepackage{bm}
\usepackage{url}
\usepackage{array}
\usepackage{color}
\usepackage{amsthm}
\usepackage{dsfont}
\usepackage{authblk}
\usepackage{amsmath}
\usepackage{amssymb}
\usepackage{caption}
\usepackage{amsfonts}
\usepackage{booktabs}
\usepackage{colortbl}
\usepackage{graphicx}
\usepackage{mathrsfs}
\usepackage{multicol}
\usepackage{multirow}
\usepackage{setspace}
\usepackage{stfloats}
\usepackage{textcomp}
\usepackage{verbatim}
\usepackage{cellspace}
\usepackage{enumerate}
\usepackage{enumitem}
\usepackage{makecell}

\usepackage{hyperref}
\hypersetup{
	linktoc      = all,
	colorlinks   = true,
	urlcolor     = blue,
	linkcolor    = blue,
	citecolor    = red
}

\usepackage[backend=biber, isbn=false, style=alphabetic, backref=true, doi=true, url=false, maxcitenames=10, mincitenames=10, maxalphanames=10, maxbibnames=10, minbibnames=5, minalphanames=5, defernumbers=true, sortlocale=en_US]{biblatex}
\usepackage[capitalize,noabbrev]{cleveref}
\usepackage[T1]{fontenc}
\usepackage[margin=1in]{geometry}
\usepackage[dvipsnames]{xcolor}
\usepackage[ruled,vlined,linesnumbered]{algorithm2e}
\usepackage[caption=false,font=normalsize,labelfont=sf,textfont=sf]{subfig}

\crefname{algocfline}{algorithm}{algorithms}
\Crefname{algocfline}{Algorithm}{Algorithms}

\setlist[itemize]{leftmargin=*}
\setlist[enumerate]{leftmargin=*}

\title{PageRank Centrality in Directed Graphs with Bounded In-Degree}

\author[1]{Mikkel Thorup}
\author[1]{Hanzhi Wang}
\author[2]{Zhewei Wei}
\author[2]{Mingji Yang}

\affil[1]{BARC, University of Copenhagen \authorcr
	\{mikkel2thorup, hanzhi.hzwang\}@gmail.com \vspace{0.8em}}
\affil[2]{Renmin University of China \authorcr
	\{zhewei, kyleyoung\}@ruc.edu.cn}

\date{}

\newtheorem{theorem}{Theorem}[section]
\newtheorem{lemma}[theorem]{Lemma}

\newtheorem{definition}{Definition}[section]

\def\eps{\varepsilon}
\def\tO{\widetilde{O}}
\def\tTheta{\widetilde{\Theta}}

\DeclareMathOperator{\E}{E}

\DeclareMathOperator{\polylog}{polylog}

\def\Nin{\mathcal{N}_{\mathrm{in}}}
\def\Nout{\mathcal{N}_{\mathrm{out}}}
\def\din{d_{\mathrm{in}}}
\def\dout{d_{\mathrm{out}}}
\def\Deltain{\Delta_{\mathrm{in}}}
\def\Deltaout{\Delta_{\mathrm{out}}}
\def\seteps{V_{\ge\eps}}
\def\setepsminus{V_{<\eps}}
\def\er{\hat{r}}
\def\ep{\hat{p}}
\def\r{r}
\def\p{p}

\def\vpi{\pi}
\def\tpi{\tilde{\vpi}}
\def\epi{\hat{\vpi}}
\def\vpilarge{\vpi_{\ge\iprime}}

\def\rmax{r_{\mathrm{max}}}

\newcommand{\indicator}[1]{\mathds{1}\left\{#1\right\}}

\def\bpush{\textup{\texttt{ApproxContributions}}\xspace}
\def\FastPPR{\textup{\texttt{FAST-PPR}}\xspace}
\def\BiPPR{\textup{\texttt{BiPPR}}\xspace}
\def\RBS{\textup{\texttt{RBS}}\xspace}
\def\push{\textup{\texttt{pushback}}\xspace}
\def\roundingpush{\textup{\texttt{RoundingPush}}\xspace}
\def\pushNoThreshold{\textup{\texttt{PushWithoutThreshold}}\xspace}

\def\indeg{\textup{\textsc{indeg}}\xspace}
\def\parent{\textup{\textsc{parent}}\xspace}
\def\outdeg{\textup{\textsc{outdeg}}\xspace}
\def\child{\textup{\textsc{child}}\xspace}
\def\jump{\textup{\textsc{jump}}\xspace}

\def\MonteCarlo{\textup{\texttt{MonteCarlo}}\xspace}

\def\iast{i^{\ast}}
\def\iprime{i^{\prime}}

\def\smallexpo{\gamma}

\def\randint{\mathrm{randint}}

\def\event{\mathcal{E}}

\begin{document}

\maketitle

\begin{abstract}

We study the computational complexity of locally estimating a node's PageRank centrality in a directed graph $G$.
For any node $t$, its PageRank centrality $\vpi(t)$ is defined as the probability that a random walk in $G$, starting from a uniformly chosen node, terminates at $t$, where each step terminates with a constant probability $\alpha \in (0,1)$.

To obtain a multiplicative $\big(1 \pm O(1)\big)$-approximation of $\vpi(t)$ with probability $\Omega(1)$, the previously best upper bound is $O\left(n^{1/2}\min\left\{ \Deltain^{1/2},\ \Deltaout^{1/2},\ m^{1/4} \right\}\right)$ from [Wang, Wei, Wen, Yang, STOC '24], where $n$ and $m$ denote the number of nodes and edges in $G$, and $\Deltain$ and $\Deltaout$ upper bound the in-degrees and out-degrees of $G$, respectively.
Using a refinement of the proof in the same paper, we establish a lower bound of $\Omega\left(n^{1/2}\min\left\{ \Deltain^{1/2} \big/ n^{\smallexpo},\Deltaout^{1/2} \big/ n^{\smallexpo},m^{1/4}\right\}\right)$, where $\smallexpo = \frac{1}{2} \left(2\max\left\{\log_{1/(1-\alpha)}\Deltain,1\right\}-1\right)^{-1}$.
As $\smallexpo$ only depends on $\Deltain$ and $n^{\smallexpo} = O(1)$ for $\Deltain = \Omega\left(n^{\Omega(1)}\right)$, the known upper bound is tight if we only parameterize the complexity by $n$, $m$, and $\Deltaout$.
However, there remains a gap of $\Omega\left(n^{\smallexpo}\right)$ when considering the maximum in-degree $\Deltain$, and this gap is large when $\Deltain$ is small.
In the extreme case where $\Deltain \le 1/(1 - \alpha)$, we have $\smallexpo = 1/2$, leading to a gap of $\Omega\left(n^{1/2}\right)$ between the bounds $O\left(n^{1/2}\right)$ and $\Omega(1)$.

In this paper, we present a new algorithm that achieves the above lower bound (up to logarithmic factors).
The algorithm assumes that $n$ and the bounds $\Deltain$ and $\Deltaout$ are known in advance.
Our key technique is a novel randomized backwards propagation process that only propagates selectively based on Monte Carlo estimated PageRank scores.

\end{abstract}

\section{Introduction} \label{sec:introduction}

Efficiently computing graph centralities is fundamental to modern network analysis.
Among these, PageRank centrality~\cite{page1998pagerank, brin1998anatomy} is of particular interest due to its simplicity, generality, and widespread applications~\cite{gleich2015pagerank}.
Specifically, given a directed graph $G = (V, E)$, the PageRank centrality $\vpi(t)$ of a node $t \in V$ equals the probability that an $\alpha$-discounted random walk starting from a random node in $G$ ends at $t$.
The length of an $\alpha$-discounted random walk follows a geometric distribution: at each step, the walk terminates with probability $\alpha \in (0,1)$, or otherwise moves to a uniformly chosen out-neighbor.
Throughout the paper, we assume $\alpha$ to be a constant following previous work~\cite{bressan2018sublinear, bressan2023sublinear, wang2024revisiting, wei2024approximating, yang2024efficient}.

We study the computational complexity of locally estimating the PageRank centrality $\vpi(t)$ for an arbitrary node $t$ in $G$, given an oracle access to the underlying graph $G$ (see \Cref{sec:preliminaries} for details).
Our goal is to efficiently compute an estimate $\epi(t)$ such that, with constant probability, $\epi(t)$ approximates $\vpi(t)$ within a constant relative error.
This problem has been studied for over a decade~\cite{chen2004local, andersen2007local, bar2008local, andersen2008robust, lofgren2014fast, lofgren2015bidirectional, lofgren2016personalized, bressan2018sublinear, bressan2023sublinear}.
The previously best upper bound is established in~\cite{wang2024revisiting} as
\begin{align*}
	O\left(n^{1/2} \min\left\{\Deltain^{1/2}, \ \Deltaout^{1/2}, \ m^{1/4}\right\}\right),
\end{align*}
where $n$ and $m$ denote the number of nodes and edges in $G$, and $\Deltain$ and $\Deltaout$ upper bound the in-degrees and out-degrees of $G$, respectively.

The same paper~\cite{wang2024revisiting} establishes lower bounds to show that, if we only parameterize the complexity by $n$, $m$, and $\Deltaout$, then the above upper bound is tight.
Nevertheless, if we take $\Deltain$ into consideration, \cite{wang2024revisiting} presents the previously best lower bound of
\begin{align*}
	\begin{cases}
	\Omega\left(n^{1/2} \min\left\{\Deltain^{1/2}, \ \Deltaout^{1/2}, \ m^{1/4}\right\}\right), & \text{if } \min\left\{\Deltain, \Deltaout \right\} = \Omega\left(n^{1/3}\right), \\
	\Omega\left(n^{1/2-\smallexpo'} \min\left\{\Deltain, \ \Deltaout\right\}^{1/2+\smallexpo'}\right), & \text{if } \min\left\{\Deltain, \Deltaout \right\} = o\left(n^{1/3}\right),
\end{cases}
\end{align*}
where $\smallexpo'>0$ is a constant.
If $\Deltain = \omega(1)$, then the constant $\smallexpo'$ can be arbitrarily small; but if $\Deltain = O(1)$, \cite{wang2024revisiting} only states that for each constant $\smallexpo'>0$, there exists a constant $\Deltain$ that yields the $\Omega\left(n^{1/2-\smallexpo'}\right)$ lower bound, and this lower bound does not apply to a fixed $\Deltain$.\footnote{This is clarified at the end of the main text in the full version of the paper~\cite[Proof of Theorem 7.2]{wang2024revisitinga}.}
Thus, there remains a gap between the known upper and lower bounds when considering the parameter $\Deltain$, and in particular, it remains unclear whether a polynomial improvement for the upper bound is possible when $\Deltain$ is small.

\subsection{Our Results} \label{subsec:result}

Our main result is a new upper bound for the complexity of estimating single-node PageRank, as formally stated in~\Cref{thm:main}.

\begin{theorem} \label{thm:main}
	Given an arbitrary graph $G$, with $n$, $\Deltain$, and $\Deltaout$ known in advance, there exists an algorithm that, with probability $\Omega(1)$, outputs a multiplicative $\big(1\pm O(1)\big)$-approximation of $\vpi(t)$ for any node $t$ in expected $\tO\left(n^{1/2}\min\left\{ \Deltain^{1/2} \big/ n^{\smallexpo}, \ \Deltaout^{1/2}\big / n^{\smallexpo}, \, m^{1/4}\right\}\right)$ time, where
	\begin{equation} \label{eqn:smallexpo}
		\smallexpo =
		\begin{cases}
			\frac{\log(1/(1-\alpha))}{4\log\Deltain - 2\log(1/(1-\alpha))}, & \textup{if } \Deltain > 1/(1-\alpha), \\
			1/2, & \textup{if } \Deltain \le 1/(1-\alpha).
		\end{cases}
	\end{equation}
\end{theorem}

An additional result is a matching lower bound that complements our upper bound, which is a more precise version of the lower bounds given in \cite{wang2024revisiting}.
The formal version of this result is stated as \Cref{thm:lower_bound_formal} in \Cref{sec:Delta_in_lower}.

\begin{theorem} [Lower bound revisited, informal] \label{thm:lower_bound_informal}
	In the worst case, computing a multiplicative $\big(1\pm O(1)\big)$-approximation of $\pi(t)$ with probability $\Omega(1)$ requires $\Omega\left(n^{1/2}\min\left\{ \Deltain^{1/2} \big/ n^{\smallexpo}, \ \Deltaout^{1/2}\big / n^{\smallexpo}, \, m^{1/4}\right\}\right)$ queries, where $\smallexpo$ is defined as in \Cref{eqn:smallexpo}.
\end{theorem}

Combining \Cref{thm:main} and \Cref{thm:lower_bound_informal} shows that our results are tight up to logarithmic factors and this is for any choice of parameters $n$, $m$,  $\Deltain$, and $\Deltaout$.

One can check that the quantity $\smallexpo$ given in \Cref{eqn:smallexpo} can be equivalently expressed as $\smallexpo = \frac{1}{2} \left(2\max\left\{\log_{1/(1-\alpha)}\Deltain,1\right\}-1\right)^{-1}$.
The role of $\smallexpo$ can be understood from the lower-bound perspective.
The basic idea is that if we have an algorithm with running time $T$, then with constant probability, it may miss a rooted tree with $n/T$ leaves and some root $u_*$ (this can be clearly seen in~\Cref{fig:lower_bound}).
Using the bound $\Deltain$ on maximal in-degree, the tree gets height $h=\log_{\Deltain} (n/T)$.
An $\alpha$-discounted random walk starts at one of the leaves with probability $1/T$, and if it does, it proceeds to $u^*$ with probability $(1-\alpha)^h$.
The value of $\smallexpo$ from~\Cref{eqn:smallexpo} corresponds to the value that satisfies $n^{-2\smallexpo} = (1-\alpha)^h$ (we also ensure that $\smallexpo \leq 1/2$).
On the other hand, a large enough $\Deltain$ would cause $h$ to be of constant order, which leads to $n^{\smallexpo} = (1-\alpha)^{-h/2} = O(1)$.

In the design of our upper-bound algorithm, the key idea is a randomized backward propagation process that only propagates selectively based on Monte Carlo estimated PageRank scores.
A technical overview is given at the beginning of \Cref{sec:Delta_in_upper}. 
Note that when $\Deltain$ is a constant, $\smallexpo$ is a constant in $(0,1/2]$ and our upper bound becomes $\tO\left(n^{1/2-\smallexpo}\right)$, which is a polynomial improvement over the previously best upper bound of $O\left(n^{1/2}\right)$. 
When $\Deltain$ gets smaller, $\smallexpo$ grows larger, and when $(1-\alpha)\Deltain \le 1$, our upper bound reduces to $\tO(1)$.
To understand this behavior, note that while a node's PageRank score aggregates contributions from up to $\Deltain$ in-neighbors, each contribution decays by a factor of $(1-\alpha)$ (see \Cref{eqn:iterative_pagerank} for the formal characterization).
Therefore, our $\tO(1)$ upper bound intuitively reflects that, when the decay rate of the $\alpha$-discounted random walk outweighs the PageRank accumulation mechanism, we can estimate $\vpi(t)$ in $\polylog(n)$ time.

For the proof of the lower-bound result, it is a refinement of the lower bound construction from \cite{wang2024revisiting}.
What makes the refinement interesting is that it matches our upper bound from \Cref{thm:main}.
Below we explain how our new lower bound relates to and strengthens the lower bound given in \cite{wang2024revisiting} in some interesting cases.

First consider the case where $\Deltain=\Omega\left(n^{\Omega(1)}\right)$.
In this case $\smallexpo=O(1/\log n)$ and $n^\smallexpo=O(1)$.
Then our new lower bound reduces to $\Omega\left(n^{1/2}\min\left\{\Deltaout^{1/2},m^{1/4}\right\}\right)$.
This was the lower bound stated in \cite{wang2024revisiting} for $\Deltain=\Omega\left(n^{1/3}\right)$, but now we get it for $\Deltain$ being any polynomial in $n$.
Also note that as long as $\Deltain=\omega(1)$, $\gamma$ can be arbitrarily small as $n$ approaches $\infty$.
Another interesting case is where $\Deltain = O(1)$.
In this case, our new lower bound becomes $\Omega\left(n^{1/2 - \smallexpo}\right)$, where $\smallexpo$ is a constant in $(0, 1/2]$ that depends on $\Deltain$.
In contrast, the previous lower bound is $\Omega\left(n^{1/2 - \smallexpo'}\right)$, but it requires $\Deltain$ to be determined by the constant $\smallexpo'$, and therefore does not apply to arbitrary fixed values of $\Deltain$.

\subsection{Paper Organization}

The rest of this paper is structured as follows.
\Cref{sec:preliminaries} introduces key concepts, notations, previous work, and two central tools for our algorithm design and analysis: the \push operation and martingales.
\Cref{sec:Delta_in_upper} presents the proof for our main result, \Cref{thm:main}, through the analysis of our proposed algorithm.
In \Cref{sec:Delta_in_lower}, we prove \Cref{thm:lower_bound_informal} by revisiting the previously known lower bounds.

\section{Preliminaries} \label{sec:preliminaries}

\paragraph{Notation.}
We consider a general directed graph $G=(V,E)$ with $n=|V|$ nodes and $m=|E|$ edges.
For each node $v\in V$, we use $\Nin(v)$ and $\Nout(v)$ to denote the sets of in-neighbors and out-neighbors of $v$, and use $\din(v)=|\Nin(v)|$ and $\dout(v)=|\Nout(v)|$ to denote the in-degree and out-degree of $v$, respectively.
Additionally, we use $\Deltain$ and $\Deltaout$ to denote the maximum in-degree and out-degree of $G$, respectively.
Throughout this paper, we assume that $n$ is known in advance and $\dout(v) > 0$ for all $v \in V$ to ensure that random walks are well-defined.
We use $\tO$ and $\tTheta$ to denote the soft-$O$ and soft-$\Theta$ notations, which hides $\polylog(n)$ factors.

\paragraph{$\alpha$-Discounted Random Walk and (Personalized) PageRank.}
An \textit{$\alpha$-discounted random walk} on $G$ is a random walk process that at each step, the walk terminates with probability $\alpha \in (0,1)$, or otherwise moves to a uniformly randomly chosen out-neighbor of the current node.
We assume that $\alpha$ is a constant.
The \textit{PageRank centrality} of a node $t \in V$, denoted by $\vpi(t)$, is defined as the probability that an $\alpha$-discounted random walk starting from a uniformly random source node in $V$ terminates at $t$.
For each $i \ge 0$, we also define the \textit{$i$-hop PageRank} of $t$, denoted by $\vpi_i(t)$, as the probability that an $\alpha$-discounted random walk starting from a uniformly random node in $V$ terminates at $t$ after exactly $i$ steps.
We have $\vpi(t) = \sum_{i \ge 0}\vpi_i(t)$.
Furthermore, for each pair of nodes $s,t \in V$, we use $\vpi(s,t)$ to denote the \textit{Personalized PageRank} value of $t$ with respect to $s$ (a.k.a. \textit{PageRank contribution} of $s$ to $t$), which is defined as the probability that an $\alpha$-discounted random walk from $s$ terminates at $t$.
It follows that $\vpi(t) = \frac{1}{n}\sum_{s \in V}\vpi(s,t)$.

\paragraph{Graph-Access Model.}
This paper adopts the standard arc-centric graph-access model~\cite{goldreich1998property,goldreich2002property} to define valid graph query operations, following prior work~\cite{lofgren2014fast, lofgren2015bidirectional, lofgren2016personalized, bressan2018sublinear, bressan2023sublinear, wang2024revisiting}.
Specifically, algorithms are allowed to access the underlying graph $G$ only through an oracle that supports the following five query operations: $\indeg(v)$, which returns $\din(v)$; $\outdeg(v)$, which returns $\dout(v)$; $\parent(v,i)$, which returns the $i$-th node in $\Nin(v)$; $\child(v,i)$, which returns the $i$-th node in $\Nout(v)$; $\jump()$, which returns a node in $V$ chosen uniformly at random.
Each of these operations can be performed in $O(1)$ time.

\subsection{Literature Review}

PageRank centrality estimation, in its many forms, has attracted extensive attention in the literature~\cite{chen2004local, fogaras2005towards, andersen2006local, andersen2007local,andersen2008local, andersen2008robust, lofgren2014fast, bressan2018sublinear, wei2018topppr,chen2020scalable,wang2021approximate, wang2020personalized, bressan2023sublinear, wang2024revisiting}.
Among them, a straightforward approach is Monte Carlo sampling~\cite{fogaras2005towards,avrachenkov2007monte,borgs2012sublinear,borgs2014multiscale}, which simulates $\alpha$-discounted random walks to estimate PageRank values.
By the Chernoff bound, $\Theta\big(1/\vpi(t)\big)$ independent random walk simulations suffice to obtain a multiplicative $\big(1 \pm O(1)\big)$-approximation of $\vpi(t)$ with constant probability.
However, all but $o(n)$ nodes in $G$ have PageRank scores of $\Theta(1/n)$, which leads to $\Theta(n)$ running time for Monte Carlo sampling in the worst case.

Another line of research~\cite{chen2004local, andersen2007local, andersen2008local, andersen2008robust, bar2008local, bressan2011local, bressan2013power, lofgren2013personalized, wang2020personalized, wang2021approximate} estimates $\vpi(t)$ by performing local exploration from $t$, which relies on the following recursive equality of PageRank that holds for each node $v$ in $G$:
\begin{align} \label{eqn:iterative_pagerank}
	\pi(v)=\sum_{u\in \Nin(v)}\frac{(1-\alpha) \pi(u)}{\dout(u)}+\frac{\alpha}{n}.
\end{align}
A milestone local exploration algorithm is $\bpush$~\cite{andersen2007local, andersen2008local}, which is originally proposed for estimating PageRank contributions to $t$.
$\bpush$ works by performing a sequence of $\push$ operations to explore the graph in reverse from $t$ to its in-neighbors step by step.
We will elaborate on the \push operation and its key invariant property in the next subsection.

Based on the ideas above, $\FastPPR$~\cite{lofgren2014fast} proposes to combine the backwards local exploration from $t$ with forwards Monte Carlo samplings, and this framework is subsequently simplified to the $\BiPPR$ algorithm~\cite{lofgren2016personalized}, which combines \bpush with Monte Carlo samplings based on the invariant property of \push operations.
Later, \cite{bressan2018sublinear, bressan2023sublinear} first achieve fully sublinear time in $(n+m)$ for estimating $\vpi(t)$ by combining Monte Carlo samplings with a new local exploration process that better handles high in-degree nodes.
Most recently, \cite{wang2024revisiting} shows that the simple combination used in the early $\BiPPR$ method can achieve the previously best upper bound of $O\left(n^{1/2} \min\left\{\Deltain^{1/2}, \ \Deltaout^{1/2}, \ m^{1/4}\right\}\right)$ for estimating $\vpi(t)$.
However, as discussed above, there remains a gap between this upper bound and the known lower bound if $\Deltain$ is small.

An interesting fact is that prior works for PageRank estimation treat high in-degree nodes as computational bottlenecks, and many efforts have been devoted to handling them~\cite{chen2004local, bar2008local, bressan2018sublinear, bressan2023sublinear, wang2020personalized}.
Specifically, a \push operation on node $v$ needs to scan all in-neighbors of $v$, which is regarded as inefficient.
Surprisingly, our lower bound reveals that a large $\Deltain$ is not an obstacle for single-node PageRank estimation—the simple \BiPPR algorithm already achieves optimal complexity in this regime.
Conversely, when $\Deltain$ is small, the upper bound of \BiPPR is not tight and it requires new techniques to improve the complexity.
This highlights an overlooked aspect of PageRank estimation: while bounded-degree graphs have been extensively studied in property testing~\cite{goldreich2002property,hellweg2012property}, the complexity of estimating PageRank on bounded in-degree graphs remains underexplored, leaving theoretical gaps to be addressed.

\paragraph{Other Related Work.}
There is another algorithm called \RBS~\cite{wang2020personalized} that also employs a randomized push process, but \RBS requires preprocessing of the graph and only propagates residues to a randomly chosen portion of the in-neighbors, so its problem setting and methodology differ from our new algorithm.
We also remark that PageRank computation has been extensively studied across diverse computational settings, including for undirected graphs~\cite{lofgren2015bidirectional, wang2023estimating, wang2024revisitingc}, dynamic graphs~\cite{zhang2016approximate,jayaram2024dynamic}, graph streams~\cite{sarma2008estimating,sarma2011estimating}, and MPC models~\cite{lacki2020walking}.

\subsection{The $\push$ Operation and the Invariant} \label{subsec:push_pre}

This subsection briefly describes the $\push$ operation proposed in~\cite{andersen2008local}.
Our algorithm also adopts $\push$, but in a different manner from prior work.

A $\push$ operation corresponds to simulating one step of $\alpha$-discounted random walk, but in the reverse direction.
It contains two memoization-style variables, residue $\r(v)$ and reserve $\p(v)$, for each node $v\in V$.
Initially, \bpush sets both $\r(v)$ and $\p(v)$ to $0$ for all $v \in V$, except $\r(t) = 1$.
Each $\push$ operation on node $v$ consists of the following three steps:
\begin{enumerate}
	\item $\p(v)\gets \p(v)+ \alpha\r(v)$, which transfers an $\alpha$ fraction of $v$’s residue $\r(v)$ to its reserve $\p(v)$;
	\item for every $u\in \Nin(v)$, $\r(u)\gets \r(u)+(1-\alpha)\r(v)/\dout(u)$, which propagates the remaining $(1-\alpha)$ fraction of $\r(v)$ to the residues of $v$’s in-neighbors;
	\item $\r(v)\gets 0$, which sets $\r(v)$ to zero.
\end{enumerate}
It is proved that each $\push$ operation on any node $v\in V$ preserves the following important invariant, which ensures the correctness of $\push$ for estimating PageRank contributions or $\vpi(t)$.
\begin{lemma}[\protect{Invariant, corollary of \cite[Lemma 3.3]{andersen2008local}}] \label{lem:invariant_backward}
	For the target node $t$, the $\push$ operations in the $\bpush$ algorithm maintain the following invariants:
	\begin{align*}
		\vpi(v,t) & = \p(v) + \sum_{w \in V}\vpi(v,w)\r(w), \quad \forall v \in V, \\
		\vpi(t) & = \sum_{v\in V}\p(v)/n+\sum_{v\in V}\vpi(v)\r(v).
	\end{align*}
\end{lemma}

\subsection{Martingales}

A crucial tool in our analysis of the proposed algorithm is the notion of martingales.
We focus on martingales that start with a constant value, which is defined as follows.

\begin{definition} [Martingale]
	A sequence of random variables $Y_0, Y_1, \dots, Y_N$ (where $Y_0$ only assumes a fixed value) is a \emph{martingale} with respect to another sequence of random variables $X_1, X_2, \dots , X_N$ if the following conditions hold:
	\begin{itemize}
		\item $Y_k$ is a function of $X_1, X_2, \dots, X_k$ for all $k\in [1,N]$;
		\item $\E\big[|Y_k|\big] < \infty$ for all $k\in [0,N]$;
		\item $\E\big[Y_{k} \mid X_1, X_2, \dots, X_{k-1}\big] = Y_{k-1}$ for all $k \in [1,N]$.
	\end{itemize}
\end{definition}

We will leverage the following theorem, known as Freedman's inequality, for establishing tail probability bounds for martingales.
It offers a Chernoff-type tail bound provided that the step size and the predictable quadratic variation of the martingale are bounded.

\begin{theorem} [Freedman's Inequality~\cite{freedman1975tail}; see also \protect{\cite[Theorem 1.1]{tropp2011freedman}}] \label{thm:freedman}
	Let $Y_0, Y_1, \dots, Y_N$ be a martingale with respect to $X_1, X_2, \dots, X_N$, and let
	\begin{align*}
		W_N = \sum_{k=1}^N \E\left[(Y_k - Y_{k-1})^2 \,\middle|\, X_1, X_2, \dots, X_{k-1}\right]
	\end{align*}
	be the \emph{predictable quadratic variation} of the martingale.
	If $|Y_k-Y_{k-1}| \le R$ for all $k \in [1,N]$ and $W_N \le \sigma^2$ always hold, then for any $\mu > 0$, we have
	\begin{align*}
		\Pr\big\{ |Y_N - Y_0| \ge \mu \big\} \le 2\exp\left(-\frac{\mu^2/2}{\sigma^2 + R\mu/3}\right).
	\end{align*}

\end{theorem}

We remark that we simplified Freedman's inequality to our needs, and the original form of Freedman's inequality is more general.

\section{Our Upper Bounds} \label{sec:Delta_in_upper}

This section proves \Cref{thm:main}.
To this end, we prove the following detailed version of this result.

\begin{theorem} \label{thm:final}
	With appropriate parameter settings, \roundingpush (\Cref{alg:roundingpush}) outputs an estimate $\epi(t)$ in expected
	\[
		\tO\left(n^{1/2}\min\left\{ \Deltain^{1/2} \big/ n^{\smallexpo}, \ \Deltaout^{1/2}\big / n^{\smallexpo}, \ m^{1/4} \right\}\right)
	\]
	time, where $\smallexpo$ is defined in \Cref{eqn:smallexpo}, while satisfying the following error guarantee:
	\begin{align*}
		\Pr\left\{\big|\epi(t)-\vpi(t)\big| \ge \frac{1}{2} \vpi(t) \right\} \le \frac{1}{10}.
	\end{align*}
\end{theorem}

The appropriate parameter settings in the theorem statement above will be given after we introduce the framework of our algorithm.
Before delving into details, we first present a high-level overview of our algorithm.

\subsection{Technical Overview}

Our algorithm consists of two phases: Monte Carlo sampling and randomized push.
In the first phase we do Monte Carlo experiments, but instead of estimating $\vpi(t)$ accurately, we only attempt to find the nodes with large PageRank centralities.
These large PageRank centralities can be estimated accurately with a small number of samples.

In the next randomized push phase, we do $\push$ operations with randomly rounded residues, but only for the nodes whose PageRank values could not be estimated accurately in the first phase.
For these nodes, before each $\push$ operation, if their residue is below a certain threshold $\rmax$, we randomly round the residue to either $0$ or $\rmax$.
Notably, the randomized rounding operations still maintain $\push$'s invariant in expectation, which facilitates a martingale-based analysis for the concentration properties of our estimate.
By restricting this rounding operation to nodes with small PageRank values, we can bound the difference and the predictable quadratic variation of the martingale, which enables us to apply Freedman's inequality to establish concentration guarantees for our estimate.

Also, this framework allows us the flexibility to introduce a special kind of Monte Carlo experiments tuned for the situation where $\Deltain$ is limited.
Specifically, in our Monte Carlo experiments, we simulate $\alpha$-discounted random walks until termination, and then continue each walk for a fixed number of additional steps.
This sampling scheme reduces the number of random walks needed to achieve desirable estimates of a variant of the PageRank value, which is close to the original PageRank value when $\Deltain$ is small and is thus sufficient for our purposes.

An algorithmic challenge lies in determining an appropriate stopping criterion for our randomized push phase.
Our remedy is to partition the push process into multiple rounds and argue that we can safely discard the remaining residues after $O(\log n)$ rounds.
We justify this approach by relating the expectations of the residues in our randomized push process to the residues in a deterministic push process, which decrease exponentially in the number of rounds performed.

We emphasize that although our algorithm is also a combination of Monte Carlo sampling and local exploration, its methodology is different from prior bidirectional methods~\cite{lofgren2014fast,lofgren2016personalized, bressan2018sublinear,bressan2023sublinear, wang2024revisiting}.
Specifically, all these prior methods use local exploration solely to construct a low-variance estimator of $\vpi(t)$, and then generate Monte Carlo samplings of this estimator to yield the final result.
In contrast, our algorithm reverses the roles of local exploration and Monte Carlo sampling: our randomized push process produces the final estimate, while the preceding Monte Carlo phase serves to guarantee the concentration property of the push process while minimizing its computational cost.

Given the foundational role of the push process in estimating PageRank and other types of random-walk probability, we believe that our $\roundingpush$ technique is of independent interest to related research.

\subsection{Technical Details}

Our new algorithm \roundingpush consists of a modified Monte Carlo algorithm and a randomized push process.
To introduce the modified Monte Carlo algorithm, we define $\vpilarge(v)=\sum_{i=\iprime}^{\infty}\vpi_i(v)$ for every $v\in V$, where $\iprime$ is a nonnegative integer to be specified later.
The modified Monte Carlo algorithm aims to estimate $\vpilarge(v)$ instead of $\vpi(v)$ for every $v\in V$.
To achieve this, the algorithm first simulates standard $\alpha$-discounted random walks until termination, and then continues each walk for exactly $\iprime$ additional steps.
Also, it multiplies a factor of $(1-\alpha)^{\iprime}$ to the estimate of each sampled node.
\Cref{alg:MC} gives a pseudocode for the modified Monte Carlo algorithm, where it takes another parameter $n_r$ as input and generates $n_r$ samplings.
The resultant Monte Carlo estimates are denoted as $\tpi(v)$ for every $v \in V$.

\begin{algorithm}[ht] \label{alg:MC}
	\DontPrintSemicolon
	\caption{$\MonteCarlo(n_r, \iprime)$}
	\KwIn{number of walks $n_r$, length $\iprime$}
	\KwOut{dictionary $\tpi()$ for Monte Carlo estimates}
	$\tpi() \gets$ empty dictionary with default value $0$ \;
	\For{$k$ \textup{from} $1$ \textup{to} $n_r$}{
	$v \gets \jump()$ \;
	\While{\textup{true}}{
		with probability $\alpha$ \textbf{break}\;
		$v \gets \child\Big(v,\randint\big(\outdeg(v)\big)\Big)$ \;
		\textcolor{gray}{// $\randint\big(\outdeg(v)\big)$ randomly returns an integer in $\big\{1,2,\dots,\outdeg(v)\big\}$} \;
	}
	\For{$i$ \textup{from} $1$ \textup{to} $\iprime$}{
		$v \gets \child\Big(v,\randint\big(\outdeg(v)\big)\Big)$ \;
	}
	$\tpi(v) \gets \tpi(v)+(1-\alpha)^{\iprime}/n_r$ \;
	}
	\Return $\tpi( )$ \;
\end{algorithm}

In \roundingpush, we will set $\seteps = \big\{v \in V \mid \tpi(v) \ge \eps\big\}$ and $\setepsminus = V \setminus \seteps$ based on the results of $\MonteCarlo(n_r,\iprime)$, for a specified parameter $\eps$.
Let $\event$ be the event that the results of the Monte Carlo algorithm in \Cref{alg:MC} satisfy all the following conditions:
\begin{align}
	\big|\tpi(v)-\vpilarge(v)\big| \le \frac{1}{20} \vpilarge(v) \quad & \text{for all } v \in \seteps, \label{eqn:E_1} \\
	\vpilarge(v) \ge \frac{2}{3} \eps \quad & \text{for all } v \in \seteps, \label{eqn:E_2} \\
	\vpilarge(v) \le 2 \eps \quad & \text{for all } v \in \setepsminus. \label{eqn:E_3}
\end{align}
We have the following lemma that lower bounds the probability of event $\event$.

\begin{lemma} \label{lem:mc_bound}
	If $n_r = \left\lceil\frac{3200(1-\alpha)^{\iprime}\ln(40n)}{\eps}\right\rceil$ in \Cref{alg:MC}, then we have $\Pr\{\event\} \ge \frac{19}{20}$.
\end{lemma}

The proof of this lemma is a standard application of the Chernoff and union bounds once one verifies that $\E\big[\tpi(v)\big]=\vpilarge(v)$ for all $v\in V$.
We defer the proof to \Cref{sec:deferred_proofs}.

Now we introduce our main algorithm \roundingpush, whose pseudocode is given in \Cref{alg:roundingpush}.
It takes several parameters as input and first runs the Monte Carlo phase by invoking $\MonteCarlo(n_r,\iprime)$ and determining $\seteps$ and $\setepsminus$ based on its results.
Next, the algorithm executes the randomized push phase in $L$ rounds by iterating $i$ from $0$ to $L-1$.
To do so, it uses $L+1$ reserve and residue dictionaries $\ep_i()$ and $\er_i()$ for $i \in [0,L]$, and initially only $\er_0(t) = 1$ is nonzero.
In each round, the algorithm examines all nodes in $\setepsminus$ with nonzero residues (instead of above the threshold $\rmax$), and performs a randomized rounding operation followed by a deterministic \push operation.
In the rounding operation, if the residue of the node $v$ in question is less than $\rmax$, its residue is set to $\rmax$ with probability $\er_i(v)/\rmax$, and to $0$ otherwise.
Subsequently, if $\er_i(v)>0$ (in which case it must be at least $\rmax$), the algorithm performs the typical \push operation on $v$.
The only difference is that in the propagation step, we update $\er_{i+1}(u)$ for each in-neighbor $u$ of $v$, instead of $\er_{i}(u)$.
We remark that the particular order in which the rounding and \push operations are performed across nodes within each level does not affect the outcome of the algorithm.
After the randomized push process, \roundingpush outputs the estimate $\epi(t)$ as
\[
	\epi(t) = \sum_{i=0}^{L-1}\sum_{v\in V}\big(\ep_i(v)/n+\tpi(v)\er_i(v)\big).
\]
Note that we do not include the $\er_L(v)$'s in the estimate $\epi(t)$.

\begin{algorithm}[ht] \label{alg:roundingpush}
	\DontPrintSemicolon
	\caption{$\roundingpush(t, \iprime, \eps, n_r, L, \rmax)$}
	\KwIn{target node $t$, length $\iprime$, threshold $\eps$, number of random walks $n_r$, length $L$, threshold $\rmax$}
	\KwOut{$\epi(t)$ as an estimate of $\vpi(t)$}
	$\tpi() \gets \MonteCarlo(n_r, \iprime)$ \;
	$\seteps \gets \big\{v \in V \mid \tpi(v) \ge \eps\big\}$, and $\setepsminus \gets V \setminus \seteps$ \;
	$\ep_i(),\er_i() \gets$ empty dictionaries with default value $0$ for all $i \in [0,L]$ \;
	$\er_0(t)\gets 1$ \;
	\For{$i$ \textup{from} $0$ \textup{to} $L-1$}{
		\textcolor{gray}{// the $i$-th level process} \;
		\For{\textup{each node} $v \in \setepsminus $ \textup{with} $\er_i(v)>0$}{
			\If{$\er_i(v)<\rmax$}{
				\textcolor{gray}{// rounding operation} \;
				$\er_i(v) \gets \begin{cases}
				\rmax, & \text{w.p. } \er_i(v)/\rmax, \\
				0, & \text{otherwise} \end{cases}$ \;
			}
			\If{$\er_i(v)>0$}{
				\textcolor{gray}{// $\push$ operation} \;
				$\ep_i(v) \gets \alpha \er_i(v)$ \;
				\For{\textup{each} $u\in \Nin(v)$}{
					$\er_{i+1}(u)\gets \er_{i+1}(u)+(1-\alpha)\er_i(v)/\outdeg(u)$ \;
				}
				$\er_i(v)\gets 0$\;
			}
		}
	}
	\Return $\epi(t)=\sum_{i=0}^{L-1}\sum_{v\in V}\big(\ep_i(v)/n+\tpi(v)\er_i(v)\big)$ \;
\end{algorithm}

Throughout the analysis below, we consider the following parameter settings:
\[
	\left\{
	\begin{aligned}
		\iast & =
		\begin{cases}
			\log_{(1-\alpha)\Deltain^2}\big(n/\min\left\{\Deltain,\Deltaout,\sqrt{m}\right\}\big), & \text{if } (1-\alpha)\Deltain > 1, \\
			\log_{1-\alpha}(1/n), & \text{otherwise},
		\end{cases} \\
		\iprime & = \lfloor \iast \rfloor, \\
		\eps & = \frac{30\alpha}{n} \cdot (\iprime+1) \cdot \max\left\{\big((1-\alpha)\Deltain\big)^{\iast}, 1\right\}, \\
		n_r & = \left\lceil\frac{3200(1-\alpha)^{\iprime}\ln(40n)}{\eps}\right\rceil, \\
		L & = \left\lceil\log_{1-\alpha}\frac{\alpha}{400n}\right\rceil+1, \\
		\rmax & = \frac{\vpi(t)}{5000L \eps}.
	\end{aligned}
	\right.
\]
Of course, we cannot directly set $\rmax$ as described above since $\vpi(t)$ is unknown.
However, for the time being we will assume that $\rmax$ is set to this ideal value.
Later at the end of this section, we will demonstrate how an adaptive approach to setting $\rmax$ achieves the same computational complexity as this ideal setting.
Additionally, if we keep track of the summation $\sum_{i=0}^{L-1}\sum_{v\in V}\big(\ep_i(v)/n+\tpi(v)\er_i(v)\big)$ throughout the execution of \roundingpush, we can readily simplify the algorithm by only maintaining two levels of residues.
However, for ease of description and analysis, we present and analyze the current version of \roundingpush.

Based on the settings of $\iast$, $\iprime$, and $\eps$, the following lemma gives some extra properties provided that event $\event$ happens.

\begin{lemma} \label{lem:E_bound}
	If event $\event$ holds, then $\vpi(v)-\vpilarge(v) \le \frac{1}{20} \vpilarge(v)$ for all $v \in \seteps$ and $\vpi(v) \le 3\eps$ for all $v \in \setepsminus$.
\end{lemma}

\begin{proof}
Recall that $\vpilarge(v) = \sum_{i=\iprime}^{\infty}\vpi_i(v)$ and $\vpi_i(v)$ is the probability that an $\alpha$-discounted random walk starting from a uniformly random node in $V$ terminates at $v$ in exactly $i$ steps.
The number of paths of length $i$ ending at $t$ is at most $\Deltain^i$, and for each of these paths, the probability that the traversal of an $\alpha$-discounted random walk in question exactly equals the path is at most $\frac{\alpha}{n}(1-\alpha)^i$.
Thus, we have
\begin{align*}
	\vpi(v)-\vpilarge(v) & \le \sum_{i=0}^{\iprime-1}\frac{\alpha (1-\alpha)^i \Deltain^i}{n} \le \frac{\alpha}{n} \cdot \iprime \cdot \max\left\{\big((1-\alpha)\Deltain\big)^{\iprime}, 1\right\} \\
	& \le \frac{\alpha}{n} \cdot \iprime \cdot \max\left\{\big((1-\alpha)\Deltain\big)^{\iast}, 1\right\} \le \frac{1}{30}\eps,
\end{align*}
since we set $\eps = \frac{30\alpha}{n} \cdot (\iprime+1) \cdot \max\left\{\big((1-\alpha)\Deltain\big)^{\iast}, 1\right\}$.
For any $v \in \seteps$, by \Cref{eqn:E_2} we have $\vpilarge(v) \ge \frac{2}{3} \eps$, so we further have
\begin{align*}
	\vpi(v)-\vpilarge(v) \le \frac{1}{30} \eps \le \frac{1}{30} \cdot \frac{3}{2} \vpilarge(v) \le \frac{1}{20} \vpilarge(v).
\end{align*}
On the other hand, for any $v \in \setepsminus$, by \Cref{eqn:E_3}, we have $\vpilarge(v) \le 2\eps$ and thus $\vpi(v) \le 2\eps + \frac{1}{30}\eps \le 3 \eps$.
This completes the proof.
\end{proof}

Next, to analyze the behavior of \roundingpush, we need to define some intermediate quantities that arise during its execution.
Here and after, for ease of description, we conceptually view that the algorithm performs rounding and $\push$ operations to each node $v \in V$ at each level $i \in [0, L-1]$ in some particular order, but these operations are null if $v \in \seteps$ or $\er_i(v)=0$.
For each $i \in [0, L-1]$ and $v \in V$, we use $\er^{-}_i(v)$ to denote the value of $\er_i(v)$ immediately before the rounding operation for node $v$ at level $i$, and $\er^{+}_i(v)$ to denote the value of $\er_i(v)$ immediately after that rounding operation (and before its subsequent possible $\push$ operation).

We shall also relate the expected values of the $\er^{-}_i(v)$'s to the corresponding values in a deterministic push process.
To this end, we give the pseudocode of the deterministic push algorithm without the $\rmax$ threshold in \Cref{alg:pushNoThreshold}.
It simply performs $L$ rounds of \push operations to all nodes with nonzero residues, which can also be viewed as a backwards version of Power Iteration (see, e.g., \cite{yang2024efficient}).
Similarly, for each $i \in [0,L-1]$ and $v \in V$, we define $\r^{-}_i(v)$ to be the value of $\r_i(v)$ immediately before the \push operation for node $v$ at level $i$, and we define $\r^{-}_{L}(v)$ to be the value of $\r_L(v)$ after the algorithm terminates.

\begin{algorithm}[ht] \label{alg:pushNoThreshold}
	\DontPrintSemicolon
	\caption{$\pushNoThreshold(t, L)$}
	\KwIn{target node $t$, length $L$}
	$\p_i(),\r_i()\gets$ empty dictionaries with default value $0$ for all $i \in [0,L]$ \;
	$\r_0(t)\gets 1$ \;
	\For{$i$ \textup{from} $0$ \textup{to} $L-1$}{
		\For{\textup{each node} $v \in V$ \textup{with} $\r_i(v)>0$}{
			$\p_i(v) \gets \alpha \r_i(v)$ \;
			\For{\textup{each} $u\in \Nin(v)$}{
				$\r_{i+1}(u) \gets \r_{i+1}(u)+(1-\alpha)\r_i(v)/\outdeg(u)$ \;
			}
			$\r_i(v)\gets 0$ \;
		}
	}
\end{algorithm}

\begin{lemma} \label{lem:exp}
	For each $i \in [0, L-1]$ and $v \in V$,
	\begin{align*}
		\E\left[\er^{+}_i(v) \,\middle|\, \er^{-}_i(v)\right] = \er^{-}_i(v) \quad \text{and} \quad \E\left[\er^{+}_i(v)\right] = \E\left[\er^{-}_i(v)\right] \le \r^{-}_{i}(v).
	\end{align*}
\end{lemma}

\begin{proof}
First, if $v\in \seteps$ or $\er^{-}_i(v) \notin (0,\rmax)$, the rounding operation takes no effect and thus $\er^{+}_i(v) = \er^{-}_i(v)$ and $\E\left[\er^{+}_i(v) \,\middle|\, \er^{-}_i(v)\right] = \er^{-}_i(v)$.
For $v\in \setepsminus$ and $\er^{-}_i(v) \in (0,\rmax)$, the rounding operation sets $\er^{+}_i(v)$ to $\rmax$ with probability $\er^{-}_i(v)/\rmax$ and to $0$ otherwise, which gives
\begin{align*}
	\E\left[\er^{+}_i(v) \,\middle|\, \er^{-}_i(v)\right]= \frac{\er^{-}_i(v)}{\rmax} \cdot \rmax = \er^{-}_i(v).
\end{align*}
Now, the first result in the lemma follows by combining these two cases, and applying the law of total expectation yields $\E\left[\er^{+}_i(v)\right] = \E\left[\er^{-}_i(v)\right]$ for all $i \in [0, L-1]$ and $v\in V$.

For the last part of the proof, we use induction on $i$.
For $i=0$, $\E\left[\er^{-}_i(v)\right] \le \r^{-}_{i}(v)$ trivially holds.
For $i \in [1,L-1]$, by the \push operations, we have
\begin{align*}
	\er^{-}_i(v) = \sum_{w \in \Nout(v) \cap \setepsminus} \frac{(1-\alpha)\er^{+}_{i-1}(w)}{\dout(v)} \le \sum_{w \in \Nout(v)} \frac{(1-\alpha)\er^{+}_{i-1}(w)}{\dout(v)}.
\end{align*}
Taking the expectation of both sides, we obtain
\begin{align*}
	\E\left[\er^{-}_i(v)\right] \le \sum_{w \in \Nout(v)} \frac{(1-\alpha)\E\left[\er^{-}_{i-1}(w)\right]}{\dout(v)} \le \sum_{w \in \Nout(v)} \frac{(1-\alpha)\r^{-}_{i-1}(w)}{\dout(v)},
\end{align*}
where we plugged in $\E\left[\er^{+}_{i-1}(w)\right] = \E\left[\er^{-}_{i-1}(w)\right]$ and the inductive hypothesis that $\E\left[\er^{-}_{i-1}(w)\right] \le r^{-}_{i-1}(w)$ for all $w \in V$.
Finally, by the property of the \push operations in \Cref{alg:pushNoThreshold}, the last summation equals $\r^{-}_{i}(v)$, completing the proof.
\end{proof}

Now we can establish an upper bound on the expected time complexity of \Cref{alg:roundingpush}, formalized in the following lemma.

\begin{lemma} \label{lem:complexity_roundingpush}
	The expected time complexity of \Cref{alg:roundingpush} is upper bounded by
	\begin{align*}
		O\left(n_r(\iprime+1)+\frac{n \vpi(t)}{\rmax} \cdot \min\left\{\Deltain, \Deltaout, \sqrt{m}\right\}\right).
	\end{align*}
\end{lemma}

\begin{proof}
First, the Monte Carlo phase (\Cref{alg:MC}) takes expected $O\big(n_r(\iprime+1)\big)$ time, since each random walk sampling takes expected $O(\iprime+1/\alpha)=O(\iprime+1)$ time and there are $n_r$ samplings.
For the push phase, its complexity can be upper bounded by the summation of $\din(v)$ times the actual number of \push operations performed to $v$, summed over all $v \in \setepsminus$.
Note that we can omit the rounding operations in this calculation since the cost of each rounding operation on node $v$ at level $i \ge 1$ can be charged to the \push operations that caused $\er^{-}_i(v)$ to be nonzero.
Hence, the complexity of the push phase can be upper bounded by $O\left(\sum_{i=0}^{L-1}\sum_{v \in \setepsminus}\indicator{\er^{+}_i(v) \ge \rmax} \cdot \din(v)\right)$, where $\indicator{\cdot}$ is the indicator function.
Based on this, we can bound its expected complexity by
\begin{align*}
	& \phantom{{}={}} O\left(\sum_{i=0}^{L-1}\sum_{v \in V}\E\left[\indicator{\er^{+}_i(v) \ge \rmax}\right] \cdot \din(v)\right) \\
	& \le O\left(\sum_{i=0}^{L-1}\sum_{v \in V}\E\left[\frac{\er^{+}_i(v)}{\rmax}\right] \cdot \din(v) \right) \le O\left(\frac{1}{\rmax}\sum_{i=0}^{L-1}\sum_{v \in V}\r^{-}_i(v)\din(v)\right),
\end{align*}
where the last inequality follows from the second result in \Cref{lem:exp}.

We now proceed to bound the quantity $\frac{1}{\rmax}\sum_{i=0}^{L-1}\sum_{v \in V}\r^{-}_i(v)\din(v)$ in the context of \Cref{alg:pushNoThreshold}.
By the procedure of \Cref{alg:pushNoThreshold}, we can see that for each $i \in [0,L-1]$ and $v \in V$, $\r^{-}_i(v)$ equals the value of $\p_i(v)/\alpha$ when the algorithm terminates.
Additionally, by the invariant property, we have $\sum_{i=0}^{L-1}\p_i(v) \le \vpi(v,t)$ at termination, for each $v \in V$.
Thus,
\begin{align*}
	\frac{1}{\rmax}\sum_{i=0}^{L-1}\sum_{v \in V}\r^{-}_i(v)\din(v) \le \frac{1}{\alpha\rmax}\sum_{v \in V}\vpi(v,t)\din(v).
\end{align*}
Interestingly, this precisely matches the complexity bound of the original \bpush algorithm~\cite{andersen2008local} with threshold $\rmax$.
As shown in \cite[Theorem 1.1]{wang2024revisiting}, this expression can be bounded by $O\left(\frac{n \vpi(t)}{\rmax} \cdot \min\left\{\Deltain, \Deltaout, \sqrt{m}\right\}\right)$, which finishes the proof.
\end{proof}

Next, we analyze the error guarantee of the estimate $\epi(t)$ returned by \roundingpush.
To this end, we define the following auxiliary random variable $Y$, which is a function of the random variables $\ep_i(v)$ and $\er_i(v)$ for $i \in [0,L-1]$ and $v \in V$ and will serve as a bridge between the estimate $\epi(t)$ and the exact PageRank value $\vpi(t)$:
\begin{align}
	Y = \sum_{i=0}^{L-1}\sum_{v \in V}\big( \ep_i(v)/n + \vpi(v) \er_i(v) \big). \label{eqn:Y}
\end{align}
Note that compared to $\epi(t)=\sum_{i=0}^{L-1}\sum_{v\in V}\big(\ep_i(v)/n+\tpi(v)\er_i(v)\big)$, \Cref{eqn:Y} involves the exact PageRank values $\vpi(v)$'s as opposed to the Monte Carlo estimates $\tpi(v)$'s.
For each $i \in [0,L-1]$, we let $Y^{(i)}$ denote the value of $Y$ immediately before the process at level $i$; and we let $Y^{(L)}$ denote the value of $Y$ when the algorithm terminates.
Now, the following lemma states that, to guarantee that $\epi(t)$ is close to $\vpi(t)$ with high probability, it suffices to ensure that $Y^{(L-1)}$ is close to $\vpi(t)$ with high probability, provided that event $\event$ holds.

\begin{lemma} \label{lem:reduction}
	If the following hold:
	\begin{align*}
		\Pr\left\{\left|Y^{(L-1)}-\vpi(t)\right| \ge \frac{1}{10} \vpi(t) \,\middle|\, \event \right\} \le \frac{1}{40},
	\end{align*}
	then \roundingpush guarantees that the output $\epi(t)$ satisfies
	\begin{align*}
		\Pr\left\{\big|\epi(t)-\vpi(t)\big| \ge \frac{1}{2} \vpi(t) \right\} \le \frac{1}{10}.
	\end{align*}
\end{lemma}

Before proving \Cref{lem:reduction}, we prove the following two lemmas that bounds $\left|Y^{(L)}-Y^{(L-1)}\right|$ and $\left|Y^{(L)}-\epi(t)\right|$.

\begin{lemma} \label{lem:Y_L_diff}
	$\Pr\left\{\left|Y^{(L)}-Y^{(L-1)}\right| \ge \frac{1}{10} \vpi(t)\right\} \le \frac{1}{40}.$
\end{lemma}

\begin{proof}
By the definition of $Y$ (\Cref{eqn:Y}) and the rounding and \push operations at level $L-1$, we have
\begin{align*}
	Y^{(L)}-Y^{(L-1)} & = \sum_{v \in \setepsminus}\left(\vpi(v)\left(\er^{+}_{L-1}(v)-\er^{-}_{L-1}(v)\right) + \frac{\alpha}{n} \cdot \er^{+}_{L-1}(v) - \vpi(v)\er^{+}_{L-1}(v)\right) \\
	& = \sum_{v \in \setepsminus}\left(\frac{\alpha}{n} \cdot \er^{+}_{L-1}(v) - \vpi(v)\er^{-}_{L-1}(v)\right).
\end{align*}
Therefore, $-\sum_{v \in V}\vpi(v)\er^{-}_{L-1}(v) \le Y^{(L)}-Y^{(L-1)} \le \sum_{v \in V}\frac{\alpha}{n} \cdot \er^{+}_{L-1}(v)$, which together with \Cref{lem:exp} and $\vpi(v) \ge \frac{\alpha}{n}$ for all $v \in V$ yields that
\begin{align*}
	\E\left[\left|Y^{(L)}-Y^{(L-1)}\right|\right] \le \E\left[\sum_{v \in V}\vpi(v)\er^{-}_{L-1}(v)\right] \le \sum_{v \in V}\vpi(v)\r^{-}_{L-1}(v).
\end{align*}

Now we show that, in \Cref{alg:pushNoThreshold}, $\r^{-}_{i}(v) \le (1-\alpha)^{i}$ for all $i \in [0,L-1]$ and $v \in V$.
We prove this by using induction on $i$.
For the base case where $i=0$, $\r^{-}_{0}(v) \le (1-\alpha)^{0} = 1$ trivially holds for all $v \in V$.
For $i \in [1,L-1]$, we have
\begin{align*}
	\r^{-}_{i}(v) & = \sum_{w \in \Nout(v)}\frac{(1-\alpha)\r^{-}_{i-1}(w)}{\dout(v)} \le \sum_{w \in \Nout(v)}\frac{(1-\alpha)(1-\alpha)^{i-1}}{\dout(v)} = (1-\alpha)^{i},
\end{align*}
where the inequality follows from the inductive hypothesis that $\r^{-}_{i-1}(w) \le (1-\alpha)^{i-1}$ for all $w \in V$.
This finishes the inductive proof.

Consequently, since $L=\lceil\log_{1-\alpha}\frac{\alpha}{400n}\rceil+1$, we can guarantee that
\begin{align*}
	\E\left[\left|Y^{(L)}-Y^{(L-1)}\right|\right] \le \sum_{v\in V}\vpi(v)\r^{-}_{L-1}(v) \le \sum_{v\in V}\vpi(v)(1-\alpha)^{L-1} = (1-\alpha)^{L-1} \le \frac{\alpha}{400n}.
\end{align*}
To conclude the proof, we apply Markov's inequality to derive that
\begin{equation*}
	\Pr\left\{\left|Y^{(L)}-Y^{(L-1)}\right| \ge \frac{1}{10} \vpi(t)\right\} \le \frac{\E\left[\left|Y^{(L)}-Y^{(L-1)}\right|\right]}{\frac{1}{10} \vpi(t)} \le \frac{\frac{\alpha}{400n}}{\frac{1}{10} \cdot\frac{\alpha}{n}} = \frac{1}{40}. \qedhere
\end{equation*}
\end{proof}

\begin{lemma} \label{lem:Y_L_error}
	If event $\event$ holds, then we can guarantee that
	\begin{align*}
		\left|Y^{(L)}-\epi(t)\right| \le \frac{1}{10} Y^{(L)}.
	\end{align*}
\end{lemma}

\begin{proof}
By the definitions of $Y^{(L)}$ and $\epi(t)$, we have
\begin{align*}
	\big|Y^{(L)}-\epi(t)\big| & = \left|\sum_{i=0}^{L-1}\sum_{v \in V}\big(\tpi(v)-\vpi(v)\big)\er_i(v)\right| \\
	& \le \sum_{i=0}^{L-1}\sum_{v \in V}\big|\tpi(v)-\vpi(v)\big|\er_i(v) = \sum_{i=0}^{L-1}\sum_{v \in \seteps}\big|\tpi(v)-\vpi(v)\big|\er_i(v),
\end{align*}
where the last equality holds because the algorithm guarantees that $\er_i(v)=0$ for all $v \in \setepsminus$ and $i \in [0,L-1]$ at termination.
For any $v \in \seteps$, \Cref{eqn:E_1} and \Cref{lem:E_bound} ensure that $\big|\tpi(v)-\vpilarge(v)\big| \le \frac{1}{20} \vpilarge(v)$ and $\big|\vpilarge(v)-\vpi(v)\big| \le \frac{1}{20} \vpilarge(v)$ holds, under event $\event$.
These inequalities imply that
\begin{align*}
	\big|\tpi(v)-\vpi(v)\big| \le \big|\tpi(v)-\vpilarge(v)\big| + \big|\vpilarge(v)-\vpi(v)\big| \le \frac{1}{20} \vpilarge(v) + \frac{1}{20} \vpilarge(v) \le \frac{1}{10} \vpi(v).
\end{align*}
Thus, we have
\begin{equation*}
	\big|Y^{(L)}-\epi(t)\big| \le \sum_{i=0}^{L-1}\sum_{v \in \seteps}\big|\tpi(v)-\vpi(v)\big|\er_i(v) \le \sum_{i=0}^{L-1}\sum_{v \in \seteps}\frac{1}{10}\vpi(v)\er_i(v) \le \frac{1}{10} Y^{{(L)}}. \qedhere
\end{equation*}
\end{proof}

\begin{proof}[Proof of \Cref{lem:reduction}]
We know that the conditions $\left|Y^{(L)}-Y^{(L-1)}\right| \le \frac{1}{10} \vpi(t)$, $\left|Y^{(L-1)}-\vpi(t)\right| \le \frac{1}{10} \vpi(t)$, and $\left|Y^{(L)}-\epi(t)\right| \le \frac{1}{10} Y^{(L)}$ together imply that
\begin{align*}
	\left|Y^{(L)}-\vpi(t)\right| \le \left|Y^{(L)}-Y^{(L-1)}\right| + \left|Y^{(L-1)}-\vpi(t)\right| \le \frac{1}{5}\vpi(t),
\end{align*}
and therefore further imply that
\begin{align*}
	\big|\epi(t)-\vpi(t)\big| & \le \left|\epi(t)-Y^{(L)}\right| + \left|Y^{(L)}-Y^{(L-1)}\right| + \left|Y^{(L-1)}-\vpi(t)\right| \\
	& \le \frac{1}{10} Y^{(L)} + \frac{1}{10} \vpi(t) + \frac{1}{10} \vpi(t) \\
	& \le \frac{1}{10} \left(1+\frac{1}{5}\right)\vpi(t) + \frac{1}{10}\vpi(t) + \frac{1}{10}\vpi(t) \\
	& < \frac{1}{2}\vpi(t).
\end{align*}
Additionally, by the conditions of the lemma, \Cref{lem:mc_bound}, \Cref{lem:Y_L_diff}, and \Cref{lem:Y_L_error}, along with the union bound, the probability that this result holds is at least $\Pr\{\event\}\left(1-\left(\frac{1}{40}+\frac{1}{40}\right)\right) \ge \frac{19}{20} \cdot \frac{19}{20} > \frac{9}{10}$, which completes the proof.
\end{proof}

The remainder of the proof hinges on demonstrating that $\Pr\left\{\left|Y^{(L-1)}-\vpi(t)\right| \ge \frac{1}{10} \vpi(t) \,\middle|\, \event \right\} \le \frac{1}{40}$, which is the core of our analysis.
For simplicity, in the rest of the analysis we assume that the event $\event$ holds and omit the conditioning on $\event$ in the probabilities and expectations.
We shall investigate a stochastic process that tracks the value of $Y$ throughout the process from level $0$ to $L-2$.
Before doing so, the next lemma establishes a key property that any \push operation at level $i \in [0,L-2]$ does not change the value of $Y$.
This is essentially a counterpart of the invariant property of the original deterministic push process, but here it does not hold for the \push operations at level $L-1$ because the definition of $Y$ (\Cref{eqn:Y}) does not involve the residues at level $L$.

\begin{lemma} \label{lem:invariant}
	Any \push operation at level $i \in [0,L-2]$ leaves the value of $Y$ invariant.
\end{lemma}

\begin{proof}
Consider a \push operation on node $v \in V$ at level $i \in [0,L-2]$.
The operation increments the reserve $\ep_i(v)$ by $\alpha\er^{+}_i(v)$, increments the residue $\er_{i+1}(u)$ by $(1-\alpha)\er^{+}_i(v)/\dout(u)$ for each $u \in \Nin(v)$, and decreases $\er_i(v)$ by $\er^{+}_i(v)$.
Together, the increment of $Y$ caused by these modifications equals
\begin{align*}
	\er^{+}_i(v)\left(\frac{\alpha}{n}+\sum_{u \in \Nin(v)}\vpi(u)\cdot\frac{1-\alpha}{\dout(u)}-\vpi(v)\right) = 0,
\end{align*}
where we used \Cref{eqn:iterative_pagerank}.
This shows that the value of $Y$ remains unchanged after the \push operation in question.
\end{proof}

Thus, to track the value of $Y$ throughout the operations from level $0$ to $L-2$, we only need to care about the rounding operations.
For each $i \in [0,L-2]$ and $v \in V$, we define $Y^{(i,v)+}$ to be the value of $Y$ immediately after the rounding operation on node $v$ at level $i$.
Here the plus sign ``$+$'' is an indication that we are considering the value of $Y$ after the corresponding rounding operation.
For notational convenience, we further assume that the rounding operations at each level are performed according to a fixed arbitrary ordering $v_1,\dots,v_n$ of the nodes in $V$.
By the definitions and \Cref{lem:invariant}, we have $Y^{(0)}=\vpi(t)$ and $Y^{(L-2, v_n)+}=Y^{(L-1)}$.
Now we consider the stochastic process
\begin{equation}
\begin{aligned}
	& Y^{(0)}, \\
	& Y^{(0, v_1)+}, Y^{(0, v_2)+}, \dots, Y^{(0, v_n)+}, \\
	& Y^{(1, v_1)+}, Y^{(1, v_2)+}, \dots, Y^{(1, v_n)+}, \\
	& \vdots \\
	& Y^{(L-2, v_1)+}, Y^{(L-2, v_2)+}, \dots, Y^{(L-2, v_n)+},
\end{aligned}
\label{eqn:martingale}
\end{equation}
which we denote as the stochastic process $\{Y\}$.
By \Cref{lem:exp}, each rounding operation leaves the involved residue (and thus the value of $Y$) unchanged in expectation, so $\{Y\}$ is a martingale with respect to the sequence of the random numbers generated by the rounding operations.
This implies that $\E\left[Y^{(L-1)}\right] = Y^{(0)} = \vpi(t)$.

Our next objective is to apply Freedman's inequality (\Cref{thm:freedman}) to establish an upper bound on the deviation probability of the last term of the martingale $\{Y\}$.
To do so, we need to analyze the difference sequence of the martingale.
By the procedure of the rounding operations and the definition of $Y$, the difference sequence of $\{Y\}$ is
\begin{equation*}
\begin{aligned}
	& \vpi(v_1)\big(\er^{+}_0(v_1)-\er^{-}_0(v_1)\big), \dots, \vpi(v_n)\big(\er^{+}_0(v_n)-\er^{-}_0(v_n)\big), \\
	& \vpi(v_1)\big(\er^{+}_1(v_1)-\er^{-}_1(v_1)\big), \dots, \vpi(v_n)\big(\er^{+}_1(v_n)-\er^{-}_1(v_n)\big), \\
	& \vdots \\
	& \vpi(v_1)\big(\er^{+}_{L-2}(v_1)-\er^{-}_{L-2}(v_1)\big), \dots, \vpi(v_n)\big(\er^{+}_{L-2}(v_n)-\er^{-}_{L-2}(v_n)\big).
\end{aligned}
\end{equation*}
The next lemma gives an upper bound on the absolute value of each term in this difference sequence.

\begin{lemma} \label{lem:difference_bound}
For each $i \in [0, L-2]$ and $v \in V$, $\left|\vpi(v)\big(\er^{+}_i(v)-\er^{-}_i(v)\big)\right| \le 3\eps\rmax$.
\end{lemma}

\begin{proof}
If $\er^{+}_i(v)=\er^{-}_i(v)$, the lemma trivially holds; otherwise, we must have $v \in \setepsminus$ and $\er^{-}_i(v) \in (0,\rmax)$, so $\vpi(v)\le 3\eps$ by \Cref{lem:E_bound} and $\left|\er^{+}_i(v)-\er^{-}_i(v)\right| \le \rmax$ by the procedure of the rounding operation, which leads to $\left|\vpi(v)\big(\er^{+}_i(v)-\er^{-}_i(v)\big)\right| \le 3\eps\rmax$, as desired.
\end{proof}

As the last component required to apply Freedman's inequality, we need to analyze the predictable quadratic variation of the martingale $\{Y\}$.
In fact, we will consider a stopped martingale $\{Z\}$ of $\{Y\}$ and investigate $\{Z\}$.
Specifically, we define $T$ to be the first index pair $(i,v)$ in \eqref{eqn:martingale} such that $\left|Y^{(i,v)+}-\vpi(t)\right| \ge \frac{1}{10} \vpi(t)$ (and $T$ equals $(L-2,v_n)$ if no such index pair exists), then $T$ is a stopping time of $\{Y\}$, and we further define $\{Z\}$ as the stopped martingale of $\{Y\}$ with respect to the stopping time $T$.
That is, $Z^{(i,v)+}$ equals $Y^{(i,v)+}$ if $(i,v)$ is before the stopping time $T$ in \eqref{eqn:martingale}, and equals the value of $Y$ at time $T$ otherwise.
Based on these definitions, the next lemma upper bounds the predictable quadratic variation of $\{Z\}$.

\begin{lemma} \label{lem:predictable_bound}
	The predictable quadratic variation of $\{Z\}$ is upper bounded by $4L\eps\rmax \cdot \vpi(t)$.
\end{lemma}

\begin{proof}
First, the expectation of the square of each term in the difference sequence of $\{Y\}$ conditioned on the past is of the form
\begin{align*}
	\E\left[\Big(\vpi(v)\big(\er^{+}_i(v)-\er^{-}_i(v)\big)\Big)^2 \,\middle|\, \er^{-}_i(v)\right] = \big(\vpi(v)\big)^2\E\left[\big(\er^{+}_i(v)-\er^{-}_i(v)\big)^2 \,\middle|\, \er^{-}_i(v)\right]
\end{align*}
for each $i \in [0,L-2]$ and $v \in V$.
This expectation is nonzero only if $v \in \setepsminus$ and $\er^{-}_i(v) \in (0,\rmax)$, in which case we have $\vpi(v)\le 3\eps$ by \Cref{lem:E_bound} and
\begin{align*}
	\er^{+}_i(v)-\er^{-}_i(v) =
	\begin{cases}
		\rmax-\er^{-}_i(v), & \text{w.p. } \er^{-}_i(v)/\rmax, \\
		-\er^{-}_i(v), & \text{otherwise},
	\end{cases}
\end{align*}
so the expectation above is upper bounded by
\begin{align*}
	& \phantom{{}={}} 3\eps \cdot \vpi(v)\left(\frac{\er^{-}_i(v)}{\rmax} \left(\rmax-\er^{-}_i(v)\right)^2 + \left(1-\frac{\er^{-}_i(v)}{\rmax}\right) \left(\er^{-}_i(v)\right)^2 \right) \\
	& = 3\eps \cdot \vpi(v)\left(\rmax\cdot\er^{-}_i(v)-\left(\er^{-}_i(v)\right)^2\right) \\
	& \le 3\eps\rmax\cdot\vpi(v)\er^{-}_i(v).
\end{align*}
Therefore, the predictable quadratic variation of $\{Y\}$ is upper bounded by
\begin{align}
	\sum_{i=0}^{L-2}\sum_{v \in V} \E\left[\Big(\vpi(v)\big(\er^{+}_i(v)-\er^{-}_i(v)\big)\Big)^2 \,\middle|\, \er^{-}_i(v)\right] \le 3\eps\rmax\sum_{i=0}^{L-2}\sum_{v \in V} \vpi(v) \er^{-}_i(v). \label{eqn:predictable_bound}
\end{align}
Recall that $Y^{(i)}$ is the value of $Y$ immediately before the operations at level $i$.
Thus, $\sum_{v \in V} \vpi(v) \er^{-}_i(v)$ can be upper bounded by $Y^{(i)}$ and \Cref{eqn:predictable_bound} can be further upper bounded by $3\eps\rmax\sum_{i=0}^{L-2} Y^{(i)}$.

Now consider the stopped martingale $\{Z\}$.
As $Z$ remains unchanged after the stopping time $T$ and $\big|Y-\vpi(t)\big| \le \frac{1}{10}\vpi(t)$ holds before the stopping time $T$ by definition, following the above argument, we can upper bound the predictable quadratic variation of $\{Z\}$ by
\begin{align*}
	& \phantom{{}={}} 3\eps\rmax\sum_{i=0}^{L-2} Y^{(i)}\cdot\indicator{\left|Y^{(i)}-\vpi(t)\right| \le \frac{1}{10} \vpi(t)} \\
	& \le 3\eps\rmax\sum_{i=0}^{L-2} \frac{11}{10}\vpi(t) \le 4L\eps\rmax \cdot \vpi(t),
\end{align*}
as desired.
\end{proof}

We are now ready to apply Freedman's inequality to upper bound the probability that $Y^{(L-1)}$ deviates from $\vpi(t)$.

\begin{lemma} \label{lem:Y_L_minus_1_concentration}
	The following holds:
	\begin{align*}
		\Pr\left\{\left|Y^{(L-1)}-\vpi(t)\right| \ge \frac{1}{10} \vpi(t)\right\} \le \frac{1}{40}.
	\end{align*}
\end{lemma}

\begin{proof}
By the definition of $\{Z\}$, $\left|Y^{(L-2, v_n)+}-\vpi(t)\right| \ge \frac{1}{10}\vpi(t)$ implies that $\left|Z^{(L-2, v_n)+}-\vpi(t)\right| \ge \frac{1}{10}\vpi(t)$.
Thus, by \Cref{lem:difference_bound} and \Cref{lem:predictable_bound}, applying Freedman's inequality (\Cref{thm:freedman}) with $\mu = \frac{1}{10} \vpi(t)$ and $\sigma^2 = 4L\eps\rmax \cdot \vpi(t)$ to the stopped martingale $\{Z\}$ yields that
\begin{align*}
	& \phantom{{}={}} \Pr\left\{\left|Y^{(L-1)}-\vpi(t)\right| \ge \frac{1}{10}\vpi(t)\right\} = \Pr\left\{\left|Y^{(L-2, v_n)+}-\vpi(t)\right| \ge \frac{1}{10}\vpi(t)\right\} \\
	& \le \Pr\left\{\left|Z^{(L-2,v_n)+}-\vpi(t)\right| \ge \frac{1}{10}\vpi(t)\right\} \\
	& \le 2 \exp\left(-\frac{\left(\frac{1}{10}\vpi(t)\right)^2/2}{4L\eps\rmax \cdot \vpi(t) + 3\eps\rmax \cdot \frac{1}{10}\vpi(t)/3}\right) \\
	& \le 2 \exp\left(-\frac{\vpi(t)}{1000L\eps\rmax}\right) = 2 \exp(-5) < \frac{1}{40},
\end{align*}
where in the last line we used $\rmax = \frac{\vpi(t)}{5000 L \eps}$.
\end{proof}

Finally, we establish \Cref{thm:final} by combining the above lemmas and adopting an adaptive setting of $\rmax$.

\begin{proof}[Proof of \Cref{thm:final}]
\Cref{lem:reduction} and \Cref{lem:Y_L_minus_1_concentration} together imply that \roundingpush guarantees that the output $\epi(t)$ satisfies the desired error guarantee of $\Pr\left\{\big|\epi(t)-\vpi(t)\big| \ge \frac{1}{2} \vpi(t) \right\} \le \frac{1}{10}$.
Now we plug in our settings of $\iast$, $\iprime$, $\eps$, $n_r$, $L$, and $\rmax$ to the expected time complexity of $O\left(n_r(\iprime+1) + \frac{n \vpi(t)}{\rmax} \cdot \min\left\{\Deltain, \Deltaout, \sqrt{m}\right\}\right)$ given in \Cref{lem:complexity_roundingpush} to obtain the final complexity bound.
For simplicity, we let $M = \min\left\{\Deltain, \Deltaout, \sqrt{m}\right\}$.
We examine the two cases separately based on the value of $(1-\alpha)\Deltain$.

If $(1-\alpha)\Deltain > 1$, then we have $\iast = \log_{(1-\alpha)\Deltain^2}(n/M)$.
Note that in this case we must have $\Deltain \ge 2$, so $(1-\alpha)\Deltain^2 = (1-\alpha)\Deltain\cdot\Deltain > 2$ and thus $\iprime \le \iast = O(\log n)$.
Additionally,
\begin{align*}
	\eps & = \frac{30\alpha}{n} \cdot \left(\iprime+1\right) \cdot \max\left\{\big((1-\alpha)\Deltain\big)^{\iast}, 1\right\} = \tTheta\left(\frac{\big((1-\alpha)\Deltain\big)^{\iast}}{n}\right) \\
	& = \tTheta\left(\frac{1}{n}\big((1-\alpha)\Deltain\big)^{\log_{(1-\alpha)\Deltain^2}(n/M)}\right) = \tTheta\left(\frac{1}{n} \left(\frac{n}{M}\right)^{\frac{\log\Deltain-\log(1/(1-\alpha))}{2\log\Deltain-\log(1/(1-\alpha))}}\right) \\
	& = \tTheta\left(n^{\frac{-\log\Deltain}{2\log\Deltain-\log(1/(1-\alpha))}} M^{-\frac{\log\Deltain-\log(1/(1-\alpha))}{2\log\Deltain-\log(1/(1-\alpha))}}\right) = \tTheta\left(n^{-\frac{1}{2}-\smallexpo}M^{\smallexpo-\frac{1}{2}}\right),
\end{align*}
where we used $\smallexpo = \frac{\log(1/(1-\alpha))}{4\log{\Deltain} - 2\log(1/(1-\alpha))}$.
Note that
\begin{align*}
	M^{\smallexpo} \le \Deltain^{\smallexpo} = \Deltain^{O(1/\log\Deltain)} = 2^{\log\Deltain\cdot O(1/\log\Deltain)} = O(1).
\end{align*}
Thus, we have $\eps = \tTheta\left(n^{-1/2-\smallexpo}M^{-1/2}\right)$.
We similarly calculate
\begin{align*}
	\Deltain^{\iast} = \Theta\left(\Deltain^{\log_{(1-\alpha)\Deltain^2}(n/M)}\right) = \Theta\left(\left(\frac{n}{M}\right)^{\frac{\log\Deltain}{2\log\Deltain-\log(1/(1-\alpha))}}\right) = \Theta\left(n^{\smallexpo+\frac{1}{2}}M^{-\frac{1}{2}}\right),
\end{align*}
which together with $\iprime = \lfloor \iast \rfloor$ leads to
\begin{align*}
	n_r(\iprime+1) & = \tO\left(\frac{(1-\alpha)^{\iprime}}{\eps}\right) = \tO\left(\frac{(1-\alpha)^{\iast}}{\frac{((1-\alpha)\Deltain)^{\iast}}{n}}\right) = \tO\left(\frac{n}{\Deltain^{\iast}}\right) = \tO\left(n^{\frac{1}{2}-\smallexpo}M^{\frac{1}{2}}\right).
\end{align*}
Also, since $\rmax = \Theta\left(\frac{\vpi(t)}{L\eps}\right)$ and $L = O(\log n)$, we have
\begin{align*}
	\frac{n \vpi(t)}{\rmax} \cdot \min\left\{\Deltain, \Deltaout, \sqrt{m}\right\} = O\left(n L \eps M\right) = \tO\left(n^{\frac{1}{2}-\smallexpo}\Deltain^{\frac{1}{2}}\right).
\end{align*}
Thus, the total complexity is $\tO\left(n^{1/2-\smallexpo}M^{1/2}\right)$ in this case.

If $(1-\alpha)\Deltain \le 1$, we have $\Deltain = O(1)$, $\iast = \log_{1-\alpha}(1/n) = \Theta(\log n)$, $\eps = \frac{30\alpha}{n} \cdot (\iprime+1) = \tTheta\left(\frac{1}{n}\right)$, and $\smallexpo=\frac{1}{2}$.
Consequently,
\begin{align*}
	n_r(\iprime+1) = \tO\left(\frac{(1-\alpha)^{\iprime}}{\eps}\right) = \tO\left(\frac{(1-\alpha)^{\log_{1-\alpha}(1/n)}}{1/n}\right) = \tO(1)
\end{align*}
and
\begin{align*}
	\frac{n \vpi(t)}{\rmax} \cdot \min\left\{\Deltain, \Deltaout, \sqrt{m}\right\} = O\left(n L \eps\right) = \tO\left(1\right).
\end{align*}
Thus, the total complexity is $\tO(1) = \tO\left(n^{1/2-\smallexpo}M^{1/2}\right)$ in this case.

Combined, we have proved the upper bound of
\[
	\tO\left(n^{1/2-\smallexpo}M^{1/2}\right) = \tO\left(n^{1/2}\min\left\{ \Deltain^{1/2} \big/ n^{\smallexpo}, \ \Deltaout^{1/2}\big / n^{\smallexpo}, \ m^{1/4}\big / n^{\smallexpo} \right\}\right).
\]
However, note that the upper bound in the theorem statement contains $m^{1/4}$ instead of $m^{1/4}\big / n^{\smallexpo}$.
This is because if $\sqrt{m} = \min\left\{\Deltain, \Deltaout, \sqrt{m}\right\}$, then we have $\Deltain \ge \sqrt{m} \ge \sqrt{n}$ and $n^{\smallexpo} = O(1)$.
Thus, we can replace $m^{1/4}\big / n^{\smallexpo}$ by $m^{1/4}$ to obtain the desired bound.

Finally, we describe how to adaptively set $\rmax$ to achieve the desired complexity bound, using an argument similar to \cite{bressan2018sublinear,bressan2023sublinear,wang2024revisiting}.
We set a computational budget of $\tTheta\left(n^{1/2}\min\left\{ \Deltain^{1/2} \big/ n^{\smallexpo}, \Deltaout^{1/2}\big / n^{\smallexpo}, m^{1/4} \right\}\right)$ and perform multiple executions of the randomized push process with exponentially decreasing values of $\rmax$ within this budget.
Specifically, the meta-algorithm iteratively runs the push process with $\rmax=\frac{1}{2}, \frac{1}{4}, \frac{1}{8}, \cdots$ until the total computational cost exceeds the budget, in which case it terminates and returns the output of the last invocation.
By choosing the constant factor in the budget setting properly, we can guarantee that the final choice of $\rmax$ is no larger than the ideal setting discussed above, thereby ensuring that the output of the algorithm satisfies the desired error guarantee.
This completes the whole proof of \Cref{thm:final}.
\end{proof}

\paragraph{Remark.}
When $(1-\alpha)\Deltain > 1$, although our setting of $\iast = \log_{(1-\alpha)\Deltain^2}\big(n/\min\left\{\Deltain,\Deltaout,\sqrt{m}\right\}\big)$ relies on $m$, we can modify the algorithm as follows to achieve the claimed complexity bound without knowing $m$ in advance.
The algorithm first checks if $\min\{\Deltain,\Deltaout\} \le \sqrt{n}$, and if this is true, we have $\min\{\Deltain,\Deltaout,\sqrt{m}\} = \min\{\Deltain,\Deltaout\}$ (since $m \ge n$), so we can simply set $\iast$ according to $\min\{\Deltain,\Deltaout\}$.
Otherwise, using the result in \cite{goldreich2008approximating}, we can estimate $m$ within a constant relative error with constant success probability in expected $\tO(\sqrt{n})$ time, and then we set $\iast$ according to $\Deltain$, $\Deltaout$, and the estimated value of $m$.
For sufficiently large $n$, we do this estimation only when $\min\{\Deltain,\Deltaout\} = \Omega(\sqrt{n})$, in which case the claimed bound becomes $\tO\left(n^{1/2}\min\left\{ \Deltain^{1/2}, \Deltaout^{1/2}, m^{1/4} \right\}\right)$, so the additional cost of $\tO(\sqrt{n})$ for estimating $m$ is negligible.
On the other hand, the approximation error for $m$ would cause $\iast$ to differ by at most $\log_{(1-\alpha)\Deltain^2}\big(O(1)\big)$, which does not affect the asymptotic behavior of quantities like $\big((1-\alpha)\Deltain\big)^{\iast}$ and $(1-\alpha)^{\iprime}$, so the final complexity bound remains unchanged.

\section{Lower Bounds Revisited} \label{sec:Delta_in_lower}

We state the formal version of \Cref{thm:lower_bound_informal} below and prove it.
Our proof framework follows that of \cite[Theorem 7.2]{wang2024revisitinga}, but we use simplified and refined graph constructions and explicitly calculate the exponent $\smallexpo$ as a function of $\Deltain$.

\begin{theorem} \label{thm:lower_bound_formal}
	Choose any integer $p \ge 2$ and any functions $\Deltain(n),\Deltaout(n) \in \Omega(1) \cap O(n)$ and $m(n)\in\Omega(n)\cap O\big(n\Deltain(n)\big)\cap O\big(n\Deltaout(n)\big)$.
	Consider any (randomized) algorithm $\mathcal{A}(t)$ that computes a multiplicative $\big(1 \pm O(1)\big)$-approximation of $\vpi(t)$ w.p. at least $1/p$, where $\mathcal{A}$ can only query the graph oracle to access unseen nodes and edges in the underlying graph.
	Then, for every sufficiently large $n$, there exists a graph $H$ such that: (i) $H$ contains $n$ nodes and $m$ edges, and its maximum in-degree and out-degree are $\Deltain$ and $\Deltaout$, respectively; (ii) $H$ contains a node $t$ such that $\mathcal{A}(t)$ requires $\Omega\left(n^{1/2}\min\left\{ \Deltain^{1/2} \big/ n^{\smallexpo},\ \Deltaout^{1/2} \big/ n^{\smallexpo},\ m^{1/4}\right\}\right)$ queries in expectation, where $\smallexpo$ is defined in \Cref{eqn:smallexpo}.
\end{theorem}

\begin{proof}
If $\Deltain \le 1/(1-\alpha)$, the desired lower bound becomes $\Omega(1)$, which trivially holds.
If $\min\{\Deltain, \Deltaout\} = \Omega\left(n^{1/3}\right)$, we have $n^{\smallexpo} = O(1)$ and the lower bound is proved in \cite{wang2024revisiting}.
Thus, we assume that $\Deltain > 1/(1-\alpha)$ and $\min\{\Deltain, \Deltaout\} = o\left(n^{1/3}\right)$, in which case the lower bound becomes $\Omega\left(n^{1/2-\smallexpo}\min\{\Deltain,\Deltaout\}^{1/2}\right)$.

For each sufficiently large $n$, we construct a family of graphs $\{H_i\}_{i=0}^p$ serving as hard instances for PageRank estimation.
These instances are designed to be computationally hard to distinguish for algorithm $\mathcal{A}$, yet $\mathcal{A}$ must successfully distinguish them to achieve the required approximation guarantee for $\pi(t)$.

We first describe the structure of $H_p$, as depicted in Figure~\ref{fig:lower_bound}.
There are five special disjoint node sets in $H_p$: $\{t\}$, $\mathcal{U}$, $\mathcal{V}$, $\mathcal{W}$, and $\mathcal{Y}$.
Each node in $\mathcal{V}$ is an ancestor of $t$, and the interconnection structure between $\mathcal{V}$ and $t$ will be described shortly.
We set $|\mathcal{U}|=|\mathcal{V}|$ and add an outgoing edge from each node in $\mathcal{U}$ to an exclusive node in $\mathcal{V}$.
Next, we set $|\mathcal{W}|=|\mathcal{V}|$ and add edges from $\mathcal{W}$ to $\mathcal{V}$, making each node in $\mathcal{V}$ has in-degree $d$ and each node in $\mathcal{W}$ has out-degree $d$, where $d=\min(\Deltain,\Deltaout)$.
We select a node in $\mathcal{V}$ as the special node $v_{*}$, and let $u_{*}$ be the parent of $v_{*}$ that is in $\mathcal{U}$.
Every node in $\mathcal{Y}$ is as an ancestor of $u_{*}$, and the structure between $\mathcal{Y}$ and $u_{*}$ will also be specified shortly.
We can easily add descendants to $t$ (omitted from the figure) to ensure that every node has a nonzero out-degree without affecting our argument.
To ensure that $H_i$ meets the required graph parameters, we augment it with an isolated subgraph (omitted from the figure) such that $H_p$ contains $n$ nodes and $m$ edges, and its maximum in-degree and out-degree are $\Deltain$ and $\Deltaout$, respectively.

\begin{figure}[ht]
	\centering
	\includegraphics[width=0.7\linewidth]{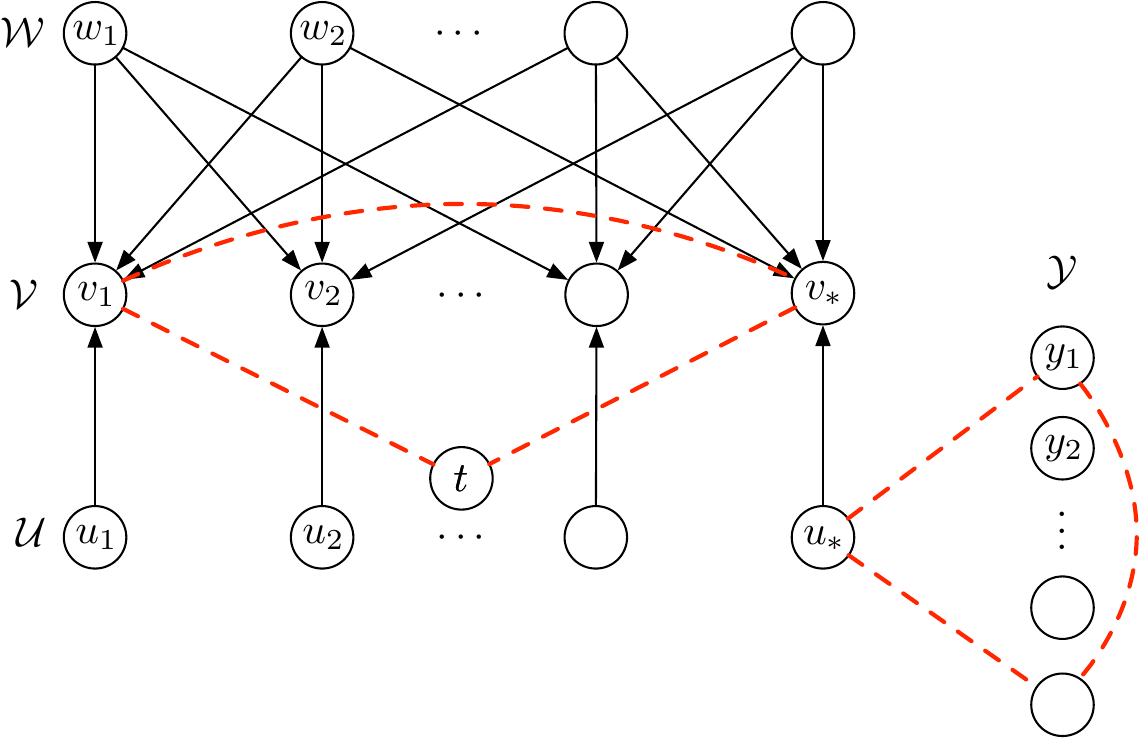}
	\caption{The graph $H_p$.
	The structures between $\mathcal{V},t$ and $\mathcal{Y},u_{\ast}$ are described below.} \label{fig:lower_bound}
\end{figure}

Figure~\ref{fig:multi_level_structure} illustrates the  structure between $\mathcal{V}$ and $t$.
The structure is analogous to a reversed full (but not necessarily complete) $\Deltain$-ary tree, where $t$ is the root and each intermediate node has $\Deltain$ parents.
The leaves (colored in \textcolor{blue}{blue}) are the node set $\mathcal{V}$.
Clearly, the depth of this structure is $\log_{\Deltain}|\mathcal{V}| \pm O(1)$, and the number of nodes in the structure is dominated by $|\mathcal{V}|$.
The structure between $\mathcal{Y}$ and $u_{\ast}$ follows the identical construction pattern, but we will set $|\mathcal{V}|$ and $|\mathcal{Y}|$ differently.

\begin{figure}[ht]
	\centering
	\includegraphics[width=0.5\linewidth]{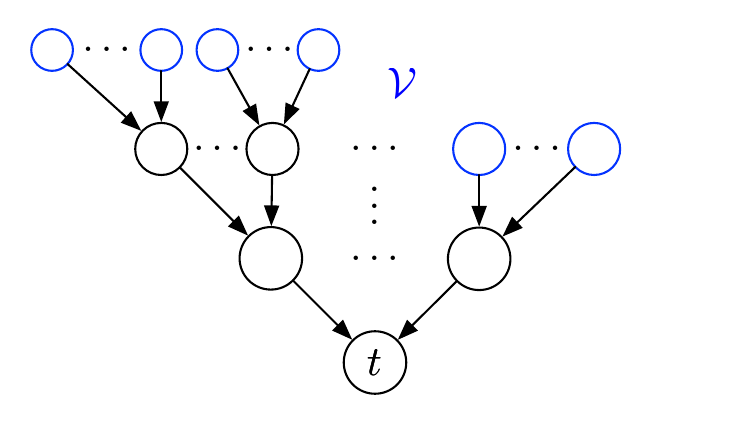}
	\caption{The structure between $\mathcal{V}$ and $t$, where the ``leaves'' (colored in \textcolor{blue}{blue}) form the node set $\mathcal{V}$ and each intermediate node has in-degree $\Deltain$.
	The structure between $\mathcal{Y}$ and $u_{\ast}$ are constructed in exactly the same way.} \label{fig:multi_level_structure}
\end{figure}

For each $0 \le i < p$, we construct $H_i$ to be identical to $H_p$ except that, $i/p\cdot|\mathcal{Y}|$ nodes in $\mathcal{Y}$ have an outgoing edge like in $H_p$, and other nodes in $\mathcal{Y}$ only have a self-loop.
Consequently, $\mathcal{A}$ must detect at least one node in the structure between $u_{\ast}$ and $\mathcal{Y}$ to distinguish among $\{H_i\}_{i=0}^{p}$, because other parts of the graphs are identical.

We now analyze the scenario where the underlying graph is uniformly sampled from $\{H_i\}_{i=0}^p$, with $v_{\ast}$ chosen uniformly from $\mathcal{V}$ and the order of the nodes and edges are permuted uniformly randomly.
We argue that under this distribution, algorithm $\mathcal{A}$ requires $\Omega\left(\min\left(d|\mathcal{V}|,n/|\mathcal{Y}|\right)\right)$ queries in expectation to distinguish $\{H_i\}_{i=0}^p$ with probability at least $1/p$.
Note that $\mathcal{A}$ can only use the $\textsc{jump}()$ operation or local exploration.
For $\textsc{jump}()$ queries, $\mathcal{A}$ needs $\Omega\left(n/|\mathcal{Y}|\right)$ attempts to detect a node in $\mathcal{Y}$ with success probability $1/p$.
Failing that, $\mathcal{A}$ must detect $u_{\ast}$ through local exploration from $t$.
To achieve this, $\mathcal{A}$ must invoke the specific query of the form $\parent(v_{*},\cdot)$ that returns $u_{\ast}$; however, among the $\Theta\left(d|\mathcal{V}|\right)$ possible queries of the form $\parent(v,\cdot)$ where $v \in \mathcal{V}$, each has equal probability of returning $u_{*}$, necessitating $\Omega\left(d|\mathcal{V}|\right)$ queries for non-negligible success probability.
The combination of these two cases gives the lower bound of $\Omega\left(\min\left(d|\mathcal{V}|,n/|\mathcal{Y}|\right)\right)$.
Now we set
\begin{align*}
	|\mathcal{V}| = \Theta\left(n^{\frac{1}{2}-\smallexpo}d^{-\frac{1}{2}+\smallexpo}\right), \quad |\mathcal{Y}| = \Theta\left(n^{\frac{1}{2}+\smallexpo}d^{-\frac{1}{2}-\smallexpo}\right).
\end{align*}
This setting gives the lower bound of $\Omega\left(\min\big(d|\mathcal{V}|,n/|\mathcal{Y}|\big)\right) = \Omega\left(n^{\frac{1}{2}-\smallexpo}\min\{\Deltain,\Deltaout\}^{\frac{1}{2}+\smallexpo}\right)$.

Lastly, we show that $\vpi^{(i)}(t) = \big(1+\Omega(1)\big)\vpi^{(i-1)}(t)$ holds for each $1\le i\le p$, where $\vpi^{(i)}(t)$ denotes the PageRank score of $t$ in $H_i$.
By the definition of PageRank, we can compute that in each $H_i$, the PageRank score of each node in $\mathcal{V} \setminus v_{\ast}$ is $\Theta(1/n)$ and $\vpi(u_{*}) = \Theta\left((1-\alpha)^{\log_{\Deltain}|\mathcal{Y}|} \cdot |\mathcal{Y}|/n\right)$.
Additionally, $\vpi^{(i)}(t) = \Theta\left((1-\alpha)^{\log_{\Deltain}|\mathcal{V}|} \left(|\mathcal{V}| + (1-\alpha)^{\log_{\Deltain}|\mathcal{Y}|} \cdot |\mathcal{Y}|\right) / n\right)$.
We also have $\vpi^{(i)}(t)-\vpi^{(i-1)}(t) = \Theta\left((1-\alpha)^{\log_{\Deltain}|\mathcal{V}|}(1-\alpha)^{\log_{\Deltain}|\mathcal{Y}|} \cdot |\mathcal{Y}|/n\right)$ for each $1\le i\le p$.
Now we can check that
\begin{align*}
	& \phantom{{}={}} (1-\alpha)^{\log_{\Deltain}|\mathcal{Y}|} \cdot |\mathcal{Y}| = |\mathcal{Y}|^{\log_{\Deltain}(1-\alpha)} \cdot |\mathcal{Y}| = |\mathcal{Y}|^{\frac{\log\Deltain-\log(1/(1-\alpha))}{\log\Deltain}} \\
	& = \Theta\left(n^{\left(\frac{1}{2}+\smallexpo\right)\cdot\frac{\log\Deltain-\log(1/(1-\alpha))}{\log\Deltain}} d^{\left(-\frac{1}{2}-\smallexpo\right)\cdot\frac{\log\Deltain-\log(1/(1-\alpha))}{\log\Deltain}}\right) \\
	& = \Theta\left(n^{\frac{\log\Deltain-\log(1/(1-\alpha))}{2\log\Deltain-\log(1/(1-\alpha))}} d^{\frac{-\log\Deltain+\log(1/(1-\alpha))}{2\log\Deltain-\log(1/(1-\alpha))}}\right) \\
	& = \Theta\left(n^{\frac{1}{2}-\smallexpo}d^{-\frac{1}{2}+\smallexpo}\right) = \Theta\big(|\mathcal{V}|\big),
\end{align*}
where we used $\smallexpo = \frac{\log(1/(1-\alpha))}{4\log\Deltain - 2\log(1/(1-\alpha))}$.
Thus, $\vpi^{(i)}(t)-\vpi^{(i-1)}(t) = \Theta\left(\vpi^{(i-1)}(t)\right)$, which ensures that $\vpi^{(i)}(t) = \big(1+\Omega(1)\big)\vpi^{(i-1)}(t)$.
Hence, to output a multiplicative $\big(1 \pm O(1)\big)$-approximation of $\vpi(t)$, $\mathcal{A}$ must distinguish the graphs in $\{H_i\}_{i=0}^p$, which requires $\Omega\left(n^{1/2-\smallexpo}\min\{\Deltain,\Deltaout\}^{1/2+\smallexpo}\right)$ queries in expectation.
Since $\min\{\Deltain,\Deltaout\}^{\smallexpo} = O\left(\Deltain^{O(1/\log\Deltain)}\right) = O(1)$, this lower bound reduces to $\Omega\left(n^{1/2-\smallexpo}\min\{\Deltain,\Deltaout\}^{1/2}\right)$, which is the desired lower bound.
\end{proof}

\paragraph{Remark.}
Recall that we provided an intuitive explanation for the value of $\smallexpo$ in \Cref{subsec:result}.
Based on the formal proof above, we can now understand the role of $\smallexpo$ more precisely.
Specifically, in our lower-bound construction, the time cost $T$ for an algorithm to find the node $u_*$ is given by
\begin{align} \label{eqn:value_T}
	T=d|\mathcal{V}|=n/|\mathcal{Y}|,
\end{align}
where here and below the equality symbol denotes asymptotic equality.
The value of $\vpi(t)$ equals
\begin{align} \label{eqn:pagerank_equal}
	\vpi(t)=|\mathcal{V}|/n=|\mathcal{Y}| (1-\alpha)^h /n,
\end{align}
where $h=\log_{\Deltain} |\mathcal{Y}|$ is the height of the rooted tree on $u_*$.
Substituting $|\mathcal{V}| = T/d$ and $|\mathcal{Y}| = n/T$ from \Cref{eqn:value_T} into \Cref{eqn:pagerank_equal}, we have
\begin{align*}
	T/d = (1-\alpha)^h n/T, 
\end{align*}
which further implies that $T=n^{1/2} d^{1/2} (1-\alpha)^{h/2}$.
Comparing this with our lower bound that $T=n^{1/2 - \smallexpo} d^{1/2}$ gives $n^{-2\smallexpo}=(1-\alpha)^{h}$.

\section{Acknowledgments}

This work was supported by the VILLUM Foundation grant 54451 and National Natural Science Foundation of China (No. U2241212).

\appendix

\section{Appendix} \label{sec:appendix}

\subsection{Deferred Proofs} \label{sec:deferred_proofs}

We will use the following version of the Chernoff bound.

\begin{theorem} [Chernoff bound, derived from \protect{\cite[Theorems 4.1 and 4.4]{chung2006survey}}] \label{thm:chernoff}
    Let $X_i$ for $1 \le i \le k$ be independent random variables such that $0 \le X_i \le r$ and $\E[X_i]=\mu$, and let $\bar{X} = \frac{1}{k}\sum_{i=1}^k X_i$.
    Then, for any $\lambda > 0$,
    \begin{align*}
        \Pr\left\{\left|\bar{X} - \mu\right| \ge \lambda\right\} \le 2\exp\left(-\frac{\lambda^2 k}{2r(\mu+\lambda/3)}\right).
    \end{align*}
\end{theorem}

\begin{proof}[Proof of \Cref{lem:mc_bound}]
For each $i \in [1,n_r]$ and $v \in V$, let $X_i(v)$ be the random variable that equals $(1-\alpha)^{\iprime}$ if the $i$-th random walk sampling in \Cref{alg:MC} terminates at $v$, and equals $0$ otherwise.
The algorithm computes each $\tpi(v)$ as $\frac{1}{n_r}\sum_{k=1}^{n_r}X_i(v)$.
Let $P_i(v)$ be the probability that a standard random walk starting from a uniformly chosen node in $V$ is at $v$ after exactly $i$ steps.
By the definition of PageRank, we have $\vpi_i(v) = \alpha(1-\alpha)^{i}P_i(v)$.
Thus, by the procedure of \Cref{alg:MC},
\begin{align*}
    \E\big[X_i(v)\big] = (1-\alpha)^{\iprime}\sum_{i=\iprime}^{\infty}\alpha(1-\alpha)^{i-\iprime} P_{i}(v) = \sum_{i=\iprime}^{\infty}\alpha(1-\alpha)^{i} P_{i}(v) = \sum_{i=\iprime}^{\infty}\vpi_i(v) = \vpilarge(v).
\end{align*}
Consequently, $\E\big[\tpi(v)\big] = \vpilarge(v)$ for each $v \in V$.

Next, we prove that with probability at least $19/20$, $\big|\tpi(v)-\vpilarge(v)\big| \le \frac{1}{20} \vpilarge(v)$ holds for all $v \in V$ with $\vpilarge(v) \ge \frac{1}{2}\eps$, and $\big|\tpi(v)-\vpilarge(v)\big| \le \frac{1}{20} \eps$ holds for all $v \in V$ with $\vpilarge(v) < \frac{1}{2}\eps$.
For $\vpilarge(v) \ge \frac{1}{2}\eps$, by the Chernoff bound (\Cref{thm:chernoff}), we have
\begin{align*}
    & \phantom{{}={}} \Pr\left\{\big|\tpi(v)-\vpilarge(v)\big| \ge \frac{1}{20} \vpilarge(v)\right\} \\
    & \le 2\exp\left(-\frac{\left(\frac{1}{20} \vpilarge(v)\right)^2 n_r}{2(1-\alpha)^{\iprime}\left(\vpilarge(v)+\frac{1}{20} \vpilarge(v)/3\right)}\right) \le 2\exp\left(-\frac{\vpilarge(v)n_r}{1600(1-\alpha)^{\iprime}}\right) \\
    & \le 2\exp\left(-\frac{\frac{1}{2}\eps n_r}{1600(1-\alpha)^{\iprime}}\right) \le \frac{1}{20n},
\end{align*}
since $n_r = \left\lceil\frac{3200(1-\alpha)^{\iprime}\ln(40n)}{\eps}\right\rceil$.
For $\vpilarge(v) < \frac{1}{2}\eps$, we similarly have
\begin{align*}
    & \phantom{{}={}} \Pr\left\{\big|\tpi(v)-\vpilarge(v)\big| \ge \frac{1}{20}\eps\right\} \\
    & \le 2\exp\left(-\frac{\left(\frac{1}{20} \eps\right)^2 n_r}{2(1-\alpha)^{\iprime}\left(\vpilarge(v)+\frac{1}{20} \eps/3\right)}\right) \le 2\exp\left(-\frac{\left(\frac{1}{20} \eps\right)^2 n_r}{2(1-\alpha)^{\iprime}\left(\frac{1}{2}\eps+\frac{1}{20} \eps/3\right)}\right) \\
    & \le 2\exp\left(-\frac{\eps n_r}{800(1-\alpha)^{\iprime}}\right) \le \frac{1}{20n}.
\end{align*}
Applying the union bound to all $v \in V$ proves the claim.

Finally, it suffices to show that the above conditions imply that event $E$ holds.
Recall that $\seteps$ is set to be $\big\{v \in V \mid \tpi(v) \ge \eps\big\}$ and $\setepsminus$ is set to be $V \setminus \seteps$.
For \Cref{eqn:E_1}, note that $v \in \seteps$ must satisfy $\vpilarge(v) \ge \frac{1}{2}\eps$ (otherwise, $\vpilarge(v) < \frac{1}{2}\eps$ and $\tpi(v) < \frac{1}{2}\eps + \frac{1}{20}\eps < \eps$), and thus we can guarantee that $\big|\tpi(v)-\vpilarge(v)\big| \le \frac{1}{20} \vpilarge(v)$ for these nodes.
For \Cref{eqn:E_2}, note that if $\vpilarge(v) < \frac{2}{3}\eps$, then $\tpi(v) < \frac{2}{3}\eps + \frac{1}{20}\eps < \eps$, so $v \in \seteps$ must satisfy $\vpilarge(v) \ge \frac{2}{3}\eps$.
For \Cref{eqn:E_3}, note that if $\vpilarge(v) > 2\eps$, then $\tpi(v) \ge \frac{19}{20}\vpilarge(v) > \eps$, so $v \in \setepsminus$ must satisfy $\vpilarge(v) \le 2\eps$.
This completes the proof.
\end{proof}

\printbibliography

@inproceedings{lacki2020walking,
  author       = {Jakub {\L}{\k{a}}cki and
                  Slobodan Mitrovi{\'c} and
                  Krzysztof Onak and
                  Piotr Sankowski},
  title        = {Walking randomly, massively, and efficiently},
  booktitle    = {Proceedings of the 52th Annual {ACM} Symposium on Theory of Computing},
  pages        = {364--377},
  year         = {2020},
  doi          = {10.1145/3357713.3384303}
}

@inproceedings{wang2024revisiting,
  author       = {Hanzhi Wang and
                  Zhewei Wei and
                  Ji{-}Rong Wen and
                  Mingji Yang},
  title        = {Revisiting local computation of {PageRank}: Simple and optimal},
  booktitle    = {Proceedings of the 56th Annual {ACM} Symposium on Theory of Computing},
  pages        = {911--922},
  year         = {2024},
  doi          = {10.1145/3618260.3649661}
}

@inproceedings{andersen2006local,
  author       = {Reid Andersen and
                  Fan R. K. Chung and
                  Kevin J. Lang},
  title        = {Local graph partitioning using {PageRank} vectors},
  booktitle    = {Proceedings of the 47th {IEEE} Symposium on Foundations of Computer Science},
  pages        = {475--486},
  year         = {2006},
  doi          = {10.1109/FOCS.2006.44}
}

@inproceedings{bressan2018sublinear,
  author       = {Marco Bressan and
                  Enoch Peserico and
                  Luca Pretto},
  title        = {Sublinear algorithms for local graph centrality estimation},
  booktitle    = {Proceedings of the 59th {IEEE} Symposium on Foundations of Computer Science},
  pages        = {709--718},
  year         = {2018},
  doi          = {10.1109/FOCS.2018.00073}
}

@inproceedings{jayaram2024dynamic,
  author       = {Rajesh Jayaram and
                  Jakub {\L}{\k{a}}cki and
                  Slobodan Mitrovi{\'c} and
                  Krzysztof Onak and
                  Piotr Sankowski},
  title        = {Dynamic {PageRank}: Algorithms and lower bounds},
  booktitle    = {Proceedings of the 51st International Colloquium on Automata, Languages, and Programming},
  volume       = {297},
  pages        = {90:1--90:19},
  year         = {2024},
  doi          = {10.4230/LIPICS.ICALP.2024.90}
}

@article{goldreich1998property,
  author       = {Oded Goldreich and
                  Shafi Goldwasser and
                  Dana Ron},
  title        = {Property testing and its connection to learning and approximation},
  journal      = {Journal of the ACM},
  volume       = {45},
  number       = {4},
  pages        = {653--750},
  year         = {1998},
  doi          = {10.1145/285055.285060}
}

@article{sarma2011estimating,
  author       = {Atish Das Sarma and
                  Sreenivas Gollapudi and
                  Rina Panigrahy},
  title        = {Estimating {PageRank} on graph streams},
  journal      = {Journal of the ACM},
  volume       = {58},
  number       = {3},
  pages        = {13:1--13:19},
  year         = {2011},
  doi          = {10.1145/1970392.1970397}
}

@article{bressan2023sublinear,
  author       = {Marco Bressan and
                  Enoch Peserico and
                  Luca Pretto},
  title        = {Sublinear algorithms for local graph-centrality estimation},
  journal      = {{SIAM} Journal on Computing},
  volume       = {52},
  number       = {4},
  pages        = {968--1008},
  year         = {2023},
  doi          = {10.1137/19M1266976}
}

@article{goldreich2002property,
  author       = {Oded Goldreich and
                  Dana Ron},
  title        = {Property testing in bounded degree graphs},
  journal      = {Algorithmica},
  volume       = {32},
  number       = {2},
  pages        = {302--343},
  year         = {2002},
  doi          = {10.1007/S00453-001-0078-7}
}

@inproceedings{wei2018topppr,
  author       = {Zhewei Wei and
                  Xiaodong He and
                  Xiaokui Xiao and
                  Sibo Wang and
                  Shuo Shang and
                  Ji{-}Rong Wen},
  title        = {{TopPPR}: Top-k {Personalized} {PageRank} queries with precision guarantees on large graphs},
  booktitle    = {Proceedings of the {ACM} International Conference on Management of Data},
  pages        = {441--456},
  year         = {2018},
  doi          = {10.1145/3183713.3196920}
}

@inproceedings{lofgren2014fast,
  author       = {Peter Lofgren and
                  Siddhartha Banerjee and
                  Ashish Goel and
                  Seshadhri Comandur},
  title        = {{FAST-PPR}: Scaling {Personalized} {PageRank} estimation for large graphs},
  booktitle    = {Proceedings of the 20th {ACM} {SIGKDD} International Conference on Knowledge Discovery and Data Mining},
  pages        = {1436--1445},
  year         = {2014},
  doi          = {10.1145/2623330.2623745}
}

@inproceedings{zhang2016approximate,
  author       = {Hongyang Zhang and
                  Peter Lofgren and
                  Ashish Goel},
  title        = {Approximate {Personalized} {PageRank} on dynamic graphs},
  booktitle    = {Proceedings of the 22nd {ACM} {SIGKDD} International Conference on Knowledge Discovery and Data Mining},
  pages        = {1315--1324},
  year         = {2016},
  doi          = {10.1145/2939672.2939804}
}

@inproceedings{wang2020personalized,
  author       = {Hanzhi Wang and
                  Zhewei Wei and
                  Junhao Gan and
                  Sibo Wang and
                  Zengfeng Huang},
  title        = {Personalized {PageRank} to a target node, revisited},
  booktitle    = {Proceedings of the 26th {ACM} {SIGKDD} International Conference on Knowledge Discovery and Data Mining},
  pages        = {657--667},
  year         = {2020},
  doi          = {10.1145/3394486.3403108}
}

@inproceedings{wang2021approximate,
  author       = {Hanzhi Wang and
                  Mingguo He and
                  Zhewei Wei and
                  Sibo Wang and
                  Ye Yuan and
                  Xiaoyong Du and
                  Ji{-}Rong Wen},
  title        = {Approximate graph propagation},
  booktitle    = {Proceedings of the 27th {ACM} {SIGKDD} International Conference on Knowledge Discovery and Data Mining},
  pages        = {1686--1696},
  year         = {2021},
  doi          = {10.1145/3447548.3467243}
}

@inproceedings{wang2024revisitingc,
  author       = {Hanzhi Wang},
  title        = {Revisiting local {PageRank} estimation on undirected graphs: Simple and optimal},
  booktitle    = {Proceedings of the 30th {ACM} {SIGKDD} International Conference on Knowledge Discovery and Data Mining},
  pages        = {3036--3044},
  year         = {2024},
  doi          = {10.1145/3637528.3671820}
}

@article{wang2023estimating,
  author       = {Hanzhi Wang and
                  Zhewei Wei},
  title        = {Estimating single-node {PageRank} in $\tilde{O}\left(\min\big\{d_t,\sqrt{m}\big\}\right)$ time},
  journal      = {Proceedings of the VLDB Endowment},
  volume       = {16},
  number       = {11},
  pages        = {2949--2961},
  year         = {2023},
  doi          = {10.14778/3611479.3611500}
}

@inproceedings{bressan2013power,
  author       = {Marco Bressan and
                  Enoch Peserico and
                  Luca Pretto},
  title        = {The power of local information in {PageRank}},
  booktitle    = {Proceedings of the 22nd International World Wide Web Conference},
  pages        = {179--180},
  year         = {2013},
  doi          = {10.1145/2487788.2487878}
}

@inproceedings{chen2004local,
  author       = {Yen{-}Yu Chen and
                  Qingqing Gan and
                  Torsten Suel},
  title        = {Local methods for estimating {PageRank} values},
  booktitle    = {Proceedings of the 13th {ACM} International Conference on Information and Knowledge Management},
  pages        = {381--389},
  year         = {2004},
  doi          = {10.1145/1031171.1031248}
}

@inproceedings{bar2008local,
  author       = {Ziv Bar{-}Yossef and
                  Li{-}Tal Mashiach},
  title        = {Local approximation of {PageRank} and reverse {PageRank}},
  booktitle    = {Proceedings of the 17th {ACM} International Conference on Information and Knowledge Management},
  pages        = {279--288},
  year         = {2008},
  doi          = {10.1145/1458082.1458122}
}

@inproceedings{bressan2011local,
  author       = {Marco Bressan and
                  Luca Pretto},
  title        = {Local computation of {PageRank}: the ranking side},
  booktitle    = {Proceedings of the 20th {ACM} International Conference on Information and Knowledge Management},
  pages        = {631--640},
  year         = {2011},
  doi          = {10.1145/2063576.2063670}
}

@inproceedings{lofgren2016personalized,
  author       = {Peter Lofgren and
                  Siddhartha Banerjee and
                  Ashish Goel},
  title        = {Personalized {PageRank} estimation and search: a bidirectional approach},
  booktitle    = {Proceedings of the 9th {ACM} International Conference on Web Search and Data Mining},
  pages        = {163--172},
  year         = {2016},
  doi          = {10.1145/2835776.2835823}
}

@article{yang2024efficient,
  author       = {Mingji Yang and
                  Hanzhi Wang and
                  Zhewei Wei and
                  Sibo Wang and
                  Ji{-}Rong Wen},
  title        = {Efficient algorithms for {Personalized} {PageRank} computation: a survey},
  journal      = {{IEEE} Transactions on Knowledge and Data Engineering},
  volume       = {36},
  number       = {9},
  pages        = {4582--4602},
  year         = {2024},
  doi          = {10.1109/TKDE.2024.3376000}
}

@article{fogaras2005towards,
  author       = {D{\'{a}}niel Fogaras and
                  Bal{\'{a}}zs R{\'{a}}cz and
                  K{\'{a}}roly Csalog{\'{a}}ny and
                  Tam{\'{a}}s Sarl{\'{o}}s},
  title        = {Towards scaling fully {Personalized} {PageRank}: Algorithms, lower bounds, and experiments},
  journal      = {Internet Mathematics},
  volume       = {2},
  number       = {3},
  pages        = {333--358},
  year         = {2005},
  doi          = {10.1080/15427951.2005.10129104}
}

@article{chung2006survey,
  author       = {Fan R. K. Chung and
                  Lincoln Lu},
  title        = {Survey: {Concentration} inequalities and martingale inequalities: a survey},
  journal      = {Internet Mathematics},
  volume       = {3},
  number       = {1},
  pages        = {79--127},
  year         = {2006},
  doi          = {10.1080/15427951.2006.10129115}
}

@article{andersen2008local,
  author       = {Reid Andersen and
                  Christian Borgs and
                  Jennifer T. Chayes and
                  John E. Hopcroft and
                  Vahab S. Mirrokni and
                  Shang{-}Hua Teng},
  title        = {Local computation of {PageRank} contributions},
  journal      = {Internet Mathematics},
  volume       = {5},
  number       = {1},
  pages        = {23--45},
  year         = {2008},
  doi          = {10.1080/15427951.2008.10129302}
}

@article{borgs2014multiscale,
  author       = {Christian Borgs and
                  Michael Brautbar and
                  Jennifer T. Chayes and
                  Shang{-}Hua Teng},
  title        = {Multiscale matrix sampling and sublinear-time {PageRank} computation},
  journal      = {Internet Mathematics},
  volume       = {10},
  number       = {1-2},
  pages        = {20--48},
  year         = {2014},
  doi          = {10.1080/15427951.2013.802752}
}

@inproceedings{andersen2007local,
  author       = {Reid Andersen and
                  Christian Borgs and
                  Jennifer T. Chayes and
                  John E. Hopcroft and
                  Vahab S. Mirrokni and
                  Shang{-}Hua Teng},
  title        = {Local computation of {PageRank} contributions},
  booktitle    = {Proceedings of the 5th International Workshop on Algorithms and Models for the Web-Graph},
  volume       = {4863},
  pages        = {150--165},
  year         = {2007},
  doi          = {10.1007/978-3-540-77004-6\_12}
}

@inproceedings{borgs2012sublinear,
  author       = {Christian Borgs and
                  Michael Brautbar and
                  Jennifer T. Chayes and
                  Shang{-}Hua Teng},
  title        = {A sublinear time algorithm for {PageRank} computations},
  booktitle    = {Proceedings of the 9th International Workshop on Algorithms and Models for the Web-Graph},
  volume       = {7323},
  pages        = {41--53},
  year         = {2012},
  doi          = {10.1007/978-3-642-30541-2\_4}
}

@inproceedings{lofgren2015bidirectional,
  author       = {Peter Lofgren and
                  Siddhartha Banerjee and
                  Ashish Goel},
  title        = {Bidirectional {PageRank} estimation: from average-case to worst-case},
  booktitle    = {Proceedings of the 12th International Workshop on Algorithms Models Web Graph},
  volume       = {9479},
  pages        = {164--176},
  year         = {2015},
  doi          = {10.1007/978-3-319-26784-5\_13}
}

@inproceedings{chen2020scalable,
  author       = {Ming Chen and
                  Zhewei Wei and
                  Bolin Ding and
                  Yaliang Li and
                  Ye Yuan and
                  Xiaoyong Du and
                  Ji{-}Rong Wen},
  title        = {Scalable graph neural networks via bidirectional propagation},
  booktitle    = {Advances in Neural Information Processing Systems 33},
  year         = {2020},
  url          = {https://proceedings.neurips.cc/paper/2020/hash/a7789ef88d599b8df86bbee632b2994d-Abstract.html}
}

@article{lofgren2013personalized,
  author       = {Peter Lofgren and
                  Ashish Goel},
  title        = {Personalized {PageRank} to a Target Node},
  journal      = {CoRR},
  volume       = {abs/1304.4658},
  year         = {2013},
  doi          = {10.48550/arXiv.1304.4658}
}

@article{wang2024revisitinga,
  author       = {Hanzhi Wang and
                  Zhewei Wei and
                  Ji{-}Rong Wen and
                  Mingji Yang},
  title        = {Revisiting local computation of {PageRank}: Simple and optimal},
  journal      = {CoRR},
  volume       = {abs/2403.12648},
  year         = {2024},
  doi          = {10.48550/ARXIV.2403.12648}
}

@article{freedman1975tail,
  author       = {David A. Freedman},
  title        = {On tail probabilities for martingales},
  journal      = {The Annals of Probability},
  pages        = {100--118},
  year         = {1975},
  doi          = {10.1214/aop/1176996452}
}

@article{brin1998anatomy,
  author       = {Sergey Brin and
                  Lawrence Page},
  title        = {The anatomy of a large-scale hypertextual web search engine},
  journal      = {Computer Networks},
  volume       = {30},
  number       = {1-7},
  pages        = {107--117},
  year         = {1998},
  doi          = {10.1016/S0169-7552(98)00110-X}
}

@techreport{page1998pagerank,
  author       = {Lawrence Page and
                  Sergey Brin and
                  Rajeev Motwani and
                  Terry Winograd},
  title        = {The {PageRank} citation ranking: Bringing order to the web},
  year         = {1998},
  institution  = {Stanford University},
  url          = {http://ilpubs.stanford.edu:8090/422/}
}

@article{avrachenkov2007monte,
  author       = {Konstantin Avrachenkov and
                  Nelly Litvak and
                  Danil Nemirovsky and
                  Natalia Osipova},
  title        = {Monte {Carlo} methods in {PageRank} computation: When one iteration is sufficient},
  journal      = {{SIAM} Journal on Numerical Analysis},
  volume       = {45},
  number       = {2},
  pages        = {890--904},
  year         = {2007},
  doi          = {10.1137/050643799}
}

@inproceedings{andersen2008robust,
  author       = {Reid Andersen and
                  Christian Borgs and
                  Jennifer T. Chayes and
                  John E. Hopcroft and
                  Kamal Jain and
                  Vahab S. Mirrokni and
                  Shang{-}Hua Teng},
  title        = {Robust {PageRank} and locally computable spam detection features},
  booktitle    = {Proceedings of the 4th International Workshop on Adversarial Information Retrieval on the Web},
  pages        = {69--76},
  year         = {2008},
  doi          = {10.1145/1451983.1452000}
}

@article{goldreich2008approximating,
  author       = {Oded Goldreich and
                  Dana Ron},
  title        = {Approximating average parameters of graphs},
  journal      = {Random Structures and Algorithms},
  volume       = {32},
  number       = {4},
  pages        = {473--493},
  year         = {2008},
  doi          = {10.1002/RSA.20203}
}

@inproceedings{sarma2008estimating,
  author       = {Atish Das Sarma and
                  Sreenivas Gollapudi and
                  Rina Panigrahy},
  title        = {Estimating {PageRank} on graph streams},
  booktitle    = {Proceedings of the 27th {ACM} {SIGMOD-SIGACT-SIGART} Symposium on Principles of Database Systems},
  pages        = {69--78},
  year         = {2008},
  doi          = {10.1145/1376916.1376928}
}

@article{tropp2011freedman,
  author       = {Joel Tropp},
  title        = {Freedman's inequality for matrix martingales},
  journal      = {Electronic Communications in Probability},
  volume       = {16},
  pages        = {262 -- 270},
  year         = {2011},
  doi          = {10.1214/ECP.v16-1624}
}

@inproceedings{hellweg2012property,
  author       = {Frank Hellweg and
                  Christian Sohler},
  title        = {Property testing in sparse directed graphs: Strong connectivity and subgraph-freeness},
  booktitle    = {Proceedings of the 20th Annual European Symposium on Algorithms},
  volume       = {7501},
  pages        = {599--610},
  year         = {2012},
  doi          = {10.1007/978-3-642-33090-2\_52}
}

@article{gleich2015pagerank,
  author       = {David F. Gleich},
  title        = {PageRank beyond the web},
  journal      = {{SIAM} Review},
  volume       = {57},
  number       = {3},
  pages        = {321--363},
  year         = {2015},
  doi          = {10.1137/140976649}
}

@inproceedings{wei2024approximating,
  author       = {Zhewei Wei and
                  Ji{-}Rong Wen and
                  Mingji Yang},
  title        = {Approximating single-source {Personalized} {PageRank} with absolute error guarantees},
  booktitle    = {Proceedings of the 27th International Conference on Database Theory},
  volume       = {290},
  pages        = {9:1--9:19},
  year         = {2024},
  doi          = {10.4230/LIPICS.ICDT.2024.9}
}

\end{document}